\newif\iffullversion
\fullversiontrue

\newif\iflncsversion 
\lncsversiontrue

\iffullversion
 \iflncsversion 
  \documentclass[envcountsect,envcountsame]{llncs}
 \else
  \documentclass[letterpaper,11pt]{article}
  \usepackage[letterpaper,hmargin=1in,vmargin=1in]{geometry}
  \usepackage{fullpage,amsthm}
  \usepackage[pdfstartview=FitH,colorlinks,linkcolor=blue,citecolor=blue]{hyperref}
  \usepackage{lmodern}
 \fi
\else
 \iflncsversion
  \documentclass[envcountsect,envcountsame]{llncs}
 \else 
  \documentclass[letterpaper,11pt]{article}
  \usepackage[letterpaper,hmargin=1in,vmargin=1in]{geometry}
  \usepackage{fullpage,amsthm}
  \usepackage[pdfstartview=FitH,colorlinks,linkcolor=blue,citecolor=blue]{hyperref}
  \usepackage{lmodern}
 \fi
\fi

\usepackage{hyperref}
\hypersetup{citecolor = black}

\usepackage{color}
\usepackage{amsmath,amsfonts,amssymb}
\usepackage{url}
\usepackage{xspace,paralist}

\providecommand{\inst}[1]{}
\renewcommand{\inst}[1]{{}$^{#1}$}

\iflncsversion 
 \providecommand{\myinstitute}[1]{\institute{\renewcommand{\inst}[1]{\!\!}#1}}
\else
\theoremstyle{plain}
\newtheorem{theorem}{Theorem}[section]
\newtheorem{lemma}[theorem]{Lemma}

\newtheorem{corollary}[theorem]{Corollary}
\newtheorem{definition}[theorem]{Definition}

\newtheorem{assumption}{Assumption}

\theoremstyle{definition}

\providecommand{\myinstitute}[1]{\date{\small\renewcommand{\and}{\\}#1}}
\fi

\numberwithin{equation}{section}

\newcommand{\algfont}[1]{\ensuremath{\mathsf{#1}}}
\newcommand{\schemefont}[1]{\ensuremath{\textsc{#1}}} 
\newcommand{\advfont}[1]{\ensuremath{\mathcal{#1}}}

\newcommand{\bra}[1]{\ensuremath{\left\langle#1\right|}}
\newcommand{\ket}[1]{\ensuremath{\left|#1\right\rangle}}
\newcommand{\hilbert}{\ensuremath{\mathcal{H}}}

\newcommand{\qadv}{\ensuremath{\advfont{A}_\text{Q}}}

\newcommand{\bin}{\ensuremath{\{0,1\}}}
\newcommand{\identity}{\mathbb{I}}
\newcommand{\parag}[1]{\paragraph*{#1}}
\newcommand{\xor}{\oplus}

\newcommand{\Angle}[1]{\ensuremath{\left\langle #1\right\rangle}}
\newcommand{\adv}{\ensuremath{\mathcal{A}}}

\newcommand{\prov}{\ensuremath{\mathcal{P}}}
\newcommand{\veri}{\ensuremath{\mathcal{V}}}

\newcommand{\exec}{\ensuremath{\leftarrow}}
\newcommand{\pk}{\ensuremath{\textit{\pk}}}

\newcommand{\sk}{\ensuremath{\textit{sk}}}

\newcommand{\ident}{\ensuremath{\schemefont{IS}}}

\newcommand{\ikgen}{\ident.\kgen}
\newcommand{\hash}{\ensuremath{\schemefont{H}}}
\newcommand{\hkgen}{\hash.\kgen}
\newcommand{\heval}{\hash.\textsf{Eval}}
\newcommand{\binrep}[1]{\ensuremath{\Angle{#1}}}

\newcommand{\ccount}{\ensuremath{\mathsf{collCount}}}
\newcommand{\coll}{\ensuremath{\mathsf{Coll}}}

\newcommand{\kgen}{\ensuremath{\algfont{KGen}}}

\newcommand{\prf}{\ensuremath{\algfont{PRF}}}
\newcommand{\PRF}{\prf}
\newcommand{\myGame}{\ensuremath{\text{Game}}}
\newcommand{\Adv}{\ensuremath{\text{Adv}}}
\newcommand{\Inv}{\ensuremath{\text{Inv}}}
\newcommand{\INV}{\ensuremath{\text{INV}}}
\newcommand{\Col}{\ensuremath{\text{Col}}}

\newcommand{\START}{\ensuremath{\textsc{START}}}
\newcommand{\RAND}{\ensuremath{\textsc{RAND}}}
\newcommand{\FINISH}{\ensuremath{\textsc{FINISH}}}
\newcommand{\SIGN}{\ensuremath{\textsc{SIGN}}}
\newcommand{\INSTANCE}{\ensuremath{\textsc{INSTANCE}}}

\newcommand{\deq}{\mathrel{:=}}

\newcommand{\parens}[1]{{\left(#1\right)}}

\newcommand{\brackets}[1]{{\left[#1\right]}}
\newcommand{\magn}[1]{{\left|#1\right|}}
\newcommand{\eq}[1]{\begin{equation*}#1\end{equation*}}

\def\E{\mathop{\mathbb E}}

\def\pk{{\sf pk}}
\def\sk{{\sf sk}}

\newcommand{\abs}[1]{\left| #1 \right|}

\newcommand{\ignore}[1]{}
\newcommand{\comment}[1]{}

\newenvironment{proofof}[1]{\medskip
\noindent{\bf Proof of #1.}}{\qed\medskip}

\definecolor{emph}{rgb}{1,0.3,0}

\newenvironment{MyItemize}{
\begin{list}{$\bullet$}{
\usecounter{enumi}
\setlength{\topsep}{4pt} \setlength{\itemsep}{4pt} }} {
\end{list}
}

\newcommand{\os}{history-free}
\newcommand{\OS}{History-free}

\newcommand{\keywords}[1]{\\\par\noindent
{\small{\em Keywords\/}: #1}}

\begin{document}

\title{Random Oracles in a Quantum World}

\date{}
\author{Dan Boneh\inst{1} \and {\"O}zg{\"u}r Dagdelen\inst{2} \and Marc Fischlin\inst{2} \and
 \\Anja Lehmann\inst{3} \and Christian Schaffner\inst{4} \and Mark Zhandry\inst{1}}

\myinstitute{\inst{1}Stanford University, USA \and \inst{2}CASED \&
  Darmstadt University of Technology, Germany \and \inst{3}IBM
  Research Zurich, Switzerland\and \inst{4}University of Amsterdam and
  CWI, The Netherlands}
\maketitle

\begin{abstract}
The interest in post-quantum cryptography --- classical systems that remain secure in the presence of a quantum adversary --- has generated elegant proposals for new cryptosystems.   Some of these systems are set in the random oracle model and are proven secure relative to adversaries that have classical access to the random oracle.  We argue that to prove post-quantum security one needs to prove security in the {\em quantum-accessible} random oracle model where the adversary can query the random oracle with quantum states.

We begin by separating the classical and quantum-accessible random oracle models by presenting  a scheme that is secure when the adversary is given classical access to the random oracle, but is insecure when the adversary can make quantum oracle queries.
We then set out to develop generic conditions under which a {\em classical} random oracle proof implies security in the {\em quantum-accessible} random oracle model.   We introduce the
concept of a {\em \os\ reduction} which is a category of classical random oracle reductions that
basically determine oracle answers independently of the history of previous queries, and we prove
that such reductions imply security in the quantum model.
We then show that certain post-quantum proposals, including ones based on lattices, can be proven secure using \os\ reductions and are therefore post-quantum secure.  We conclude with a rich set of open problems in this area.
\keywords{Quantum, Random Oracle, Signatures, Encryption}
\end{abstract}


\section{Introduction}

The threat to existing public-key systems posed by quantum
computation~\cite{Shor} has generated considerable interest in {\em
  post-quantum} cryptosystems, namely systems that remain secure in
the presence of a quantum adversary.  A promising direction is
lattice-based cryptography, where the underlying problems are related
to finding short vectors in high dimensional lattices.  These problems
have so far remained immune to quantum attacks and some evidence
suggests that they may be hard for quantum computers~\cite{Regev02}.

As it is often the case, the most efficient constructions in
lattice-based cryptography are set in the random oracle (RO)
model~\cite{BR1}.  For example, Gentry, Peikert, and
Vaikuntanathan~\cite{Gentry2008} give elegant random oracle model
constructions for existentially unforgeable signatures and for
identity-based encryption. Gordon, Katz, and
Vaikuntanathan~\cite{Gordon2010} construct a random oracle model group
signature scheme. Boneh and Freeman~\cite{Boneh2011} give a random
oracle homomorphic signature scheme and Cayrel et al.~\cite{Cayrel2010}
give a lattice-based signature scheme using the Fiat-Shamir random oracle
heuristic.  Some of these lattice constructions can now be realized
without random oracles, but at a significant cost in
performance~\cite{CHKP10,ABB1,Boy10}.

\paragraph{\bf Modeling Random Oracles for Quantum Attackers.}
While quantum resistance is good motivation for lattice-based
constructions, most random oracle systems to date are only proven secure
relative to an adversary with {\em classical} access to the random
oracle. In this model the adversary is given oracle access to a
random hash function $O:\bin^*\to\bin^*$ and it can only ``learn'' a value $O(x)$ by
querying the oracle $O$ at the classical state~$x$.  However, to obtain a concrete system, the random oracle is eventually replaced by a concrete hash function thereby enabling a quantum attacker to evaluate this hash function on \emph{quantum states}.
To capture this issue in the model, we allow the adversary to evaluate the random oracle ``in superposition'', that is, the adversary
can submit quantum states $\ket{\varphi}=\sum \alpha_x\ket{x}$ to the oracle $O$ and receives back the evaluated state $\sum\alpha_x
\ket{O(x)}$ (appropriately encoded to make the transformation unitary). We call this the \emph{quantum(-accessible) random oracle model}. It complies with similar efforts from learning theory~\cite{BshJac99,SerGor04} and computational complexity~\cite{Bennett1997} where oracles are quantum-accessible, and from lower bounds for quantum collision finders~\cite{AaSh04}.
Still, since we are only interested in classical cryptosystems, \emph{honest} parties and the scheme's algorithms can access $O$ only via classical bit strings.


\smallskip
Proving security in the quantum-accessible RO model is considerably harder than in the classical model.   As a simple example, consider the case of digital signatures. A standard proof strategy in the classical settings is to choose randomly one of the adversary's RO queries and embed in the response a given instance of a challenge problem.   One then hopes that the adversary uses this response in his signature forgery.  If the adversary makes~$q$ random oracle queries, then this happens with probability $1/q$ and since~$q$ is polynomial this success probability is sufficiently high for the proof of security in the classical setting.   Unfortunately, this strategy fails completely in the quantum-accessible random oracle model since {\em every} random oracle query potentially evaluates the random oracle at exponentially many points.   Therefore, embedding the challenge in one response will be of no use to the reduction algorithm.    This simple example shows that proving security in the classical RO model does not necessarily prove post-quantum security.

More abstractly, the following common classical proof techniques are not known to carry over to the quantum settings offhand:
\begin{itemize}
\item {\normalfont Adaptive Programmability:} The classical random oracle model allows a simulator to program the answers of the random oracle for an adversary, often adaptively. 
Since the quantum adversary can query the random oracle  with a state in superposition, the adversary may get some information about all exponentially many values right at the beginning, thereby making it difficult to program the oracle adaptively.
\smallskip

\item {\normalfont Extractability/Preimage Awareness:} Another application of the random oracle model for classical adversaries is that the simulator learns the pre-images the adversary
  is interested in. This is, for example, crucial to simulate decryption queries in the security proof for OAEP~\cite{FuOkPoSt2001}. For quantum-accessible
  oracles the actual query may be hidden in a superposition of exponentially many states, and it is unclear how to extract the right query.

 \smallskip

\item {\normalfont Efficient Simulation:} In the classical world, we can simulate an exponential-size random oracle efficiently via lazy sampling: simply pick random but
consistent answers ``on the fly".  With quantum-accessible random oracles the adversary can evaluate the random oracle on all inputs simultaneously, making it harder to apply the on-demand strategy for classical oracles.
\smallskip

\item {\normalfont Rewinding/Partial Consistency:}  Certain random oracle proofs~\cite{PoiSte00} require rewinding the adversary, replaying some hash values but changing at least a single value. Beyond the usual problems of rewinding quantum adversaries, we again encounter the fact that we may not be able to change hash values unnoticed.
We note that some form of rewinding is possible for quantum zero-knowledge~\cite{Wat09}.
\smallskip

\end{itemize}

\noindent
We do not claim that these problems are insurmountable.  In fact, we show how to resolve the issue of efficient simulation by using (quantum-accessible) pseudorandom functions. These are pseudorandom functions where the quantum distinguisher can  submit quantum states to the pseudorandom or random oracle. By this technique, we can efficiently simulate the quantum-accessible random oracle through the (efficient) pseudorandom function.
While pseudorandom functions where the distinguisher may use quantum power but only gets classical access to the function can be derived from quantum-immune pseudorandom generators~\cite{GolGolMic86}, it is an open problem if the stronger quantum-accessible pseudorandom functions exist.

Note, too, that we do not seek to solve the problems related to the random oracle model which appear already in the classical settings~\cite{CanGolHal98}. Instead we show that for post-quantum security one should allow for quantum access to the random oracle in order to capture attacks that are available when the hash function is eventually instantiated.

\subsection{Our Contributions}

\paragraph{\bf Separation.}
We begin with a separation between the classical and quantum-accessible RO models by presenting a two-party protocol which is:
\begin{itemize}
 \item secure in the classical random oracle model,
 \item secure against quantum attackers with classical access to
  the random oracle model, but insecure under \emph{any} implementation
  of the hash function, and
 \item insecure in the quantum-accessible random oracle model.
\end{itemize}
The protocol itself assumes that (asymptotically) quantum computers are  faster than classical (parallel) machines and uses the quadratic gap due to Grover's algorithms and its application to collision search~\cite{BraHoyTap98} to separate secure from insecure executions.

\paragraph{\bf Constructions.} Next, we set out to give general conditions under which a {\em classical} RO proof implies security
for a {\em quantum} RO.   Our goal is to provide generic tools by which authors can simply state that their classical proof has the ``right'' structure and therefore their proof implies quantum security.   We give two flavors of results:
\begin{itemize}
\item For signatures, we define a proof structure we call a {\em \os\ reduction} which roughly says that the reduction answers oracle queries independently of the history of queries. We prove that any classical proof that happens to be a \os\ reduction implies quantum existential unforgeability for the signature scheme.   We then show that the GPV random oracle signature scheme~\cite{Gentry2008} has a \os\ reduction and is therefore secure in the quantum settings.

Next, we consider signature schemes built from claw-free permutations.  The first is the Full Domain Hash (FDH) signature system of Bellare and Rogaway~\cite{BR1}, for which we show that the classical proof technique due to Coron~\cite{Coron2000} is \os.  We also prove the quantum security of a variant of FDH due to Katz and Wang~\cite{Katz2003} which has a tight security reduction.  Lastly, we note that, as observed in~\cite{Gentry2008}, claw-free permutations give rise to preimage sampleable trapdoor functions, which gives another FDH-like signature scheme with a tight security reduction.  In all three cases the reductions in the quantum-accessible random oracle model achieve essentially the same tightness as their classical analogs.

Interestingly, we do not know of a \os\ reduction for the generic Full Domain Hash of Bellare and Rogaway~\cite{BR1}.  One reason is that proofs for generic FDH must somehow program the random oracle, as shown in \cite{NPRO}.  We leave the quantum security of generic FDH as an interesting open problem.  It is worth noting that at this time the quantum security of FDH is somewhat theoretical since we have no candidate quantum-secure trapdoor permutation to instantiate the FDH scheme, though this may change once a candidate is proposed.

\smallskip

\item For encryption we prove the quantum CPA security of an encryption scheme due to Bellare and Rogaway~\cite{BR1} and the quantum CCA security of a hybrid encryption variant of \cite{BR1}. \end{itemize}

\noindent
Many open problems remain in this space.  For signatures, it is still open to prove the quantum security of signatures that result from applying the Fiat-Shamir heuristic to a~$\Sigma$ identification protocol, for example, as suggested in~\cite{Cayrel2010}.   Similarly, proving security of generic FDH is still open.   For CCA-secure encryption, it is unknown if generic CPA to CCA transformations, such as~\cite{FO99}, are secure in the quantum settings.  Similarly, it is not known if lattice-based identity-based encryption systems secure in the classical RO model (e.g.~as in~\cite{Gentry2008,ABB2}) are also secure in the quantum random oracle model.

\paragraph{\bf Related Work.}
%
The quantum random oracle model has been used in a few previous
constructions.  Aaronson~\cite{money} uses quantum random
oracles to construct unclonable public-key quantum money.  Brassard
and Salvail~\cite{Brassard2008} give a modified version of Merkle's
Puzzles, and show that any quantum attacker must query the random
(permutation) oracle asymptotically more times than honest parties.
Recently, a modified version was proposed that restores some level of security even in
the presence of a quantum adversary~\cite{BHKKLS11}.
Quantum random oracles have also been used to prove impossibility
results for quantum computation.  For example, Bennett et
al.~\cite{Bennett1997} show that relative to a random oracle, a
quantum computer cannot solve all of NP.

Some progress toward identifying sufficient conditions under which classical protocols are also quantum immune has been made by
Unruh~\cite{Unr10} and Hallgren et al.~\cite{hasmso11}.
These results show that, if a cryptographic protocol can be shown to be (computationally \cite{hasmso11} resp.~statistically \cite{Unr10}) secure in Canetti's universal composition (UC) framework \cite{Can01} against classical adversaries, then the protocol
is also resistant against (computationally bounded resp.~unbounded) quantum adversaries. This, however, means that the underlying protocol must
already provide strong security guarantees in the first place, namely,
universal composition security, which is typically more than the
aforementioned schemes in the literature satisfy. This also applies
to similar results by Hallgren et al.~\cite{hasmso11} for
so-called simulation-based security notions for the starting protocol.
Furthermore, all these results do not seem to be applicable immediately
to the random oracle model where the quantum adversary now has \emph{quantum} access
to the random function (but where the ideal functionality for the random oracle in the UC framework
would have only been defined for classical access according to the classical protocol specification),
and where the question of instantiation is an integral step which needs to be considered.

\section{Preliminaries}\label{prelim}

A non-negative function $\epsilon=\epsilon(n)$ is negligible if, for all polynomials $p(n)$ we have that $\epsilon(n)<p(n)^{-1}$ for all sufficiently large $n$.  The variational distance between two distributions $D_1$ and $D_2$ over $\Omega$ is given by \eq{|D_1-D_2|=\sum_{x\in\Omega}\magn{\Pr[x|D_1]-\Pr[x|D_2]}.}  If the distance between two output distributions is $\epsilon$, the difference in probability of the output satisfying a certain property is at most $\epsilon$.

A classical randomized algorithm $A$ can be thought of in two ways.  In the first, $A$ is given an input $x$, $A$ makes some coin tosses during its computation, and ultimately outputs some value $y$.  We denote this action by  $A(x)$ where $A(x)$ is a random variable. Alternatively, we can give $A$ both its input $x$ and randomness $r$ in which case we denote this action as $A(x;r)$.  For a classical algorithm, $A(x;r)$ is deterministic. An algorithm $A$ runs in probabilistic polynomial-time (PPT) if it runs in polynomial time in the security parameter
(which we often omit from the input for sake of simplicity).

\subsection{Quantum Computation}

We briefly give some background on quantum computation and refer to \cite{ChaNie00} for a more complete discussion.
A quantum system $A$ is associated to a (finite-dimensional) complex Hilbert space $\hilbert_A$ with an inner product $\Angle{\cdot | \cdot}$. The state of the system is described by a vector $\ket{\varphi}\in\hilbert_A$ such that the Euclidean norm $\|\ket{\varphi}\|=\sqrt{\Angle{\varphi | \varphi}}$ is $1$.
Given quantum systems $A$ and $B$ over spaces $\hilbert_A$ and $\hilbert_B$, respectively, we define the joint or composite quantum system through the tensor product $\hilbert_A\otimes\hilbert_B$.  The product state of $\ket{\varphi_A}\in\hilbert_A$ and $\ket{\varphi_B}\in\hilbert_B$ is denoted by
$\ket{\varphi_A}\otimes\ket{\varphi_B}$ or simply $\ket{\varphi_A}\ket{\varphi_B}$.  An $n$-qubit system lives in the joint quantum system of $n$ two-dimensional Hilbert spaces.  The standard orthonormal computational basis $\ket{x}$ for such a system is given by $\ket{x_1}\otimes\dots\otimes\ket{x_n}$ for $x=x_1\dots
x_n$. Any (classical) bit string $x$ is encoded into a quantum state as $\ket{x}$. An arbitrary pure $n$-qubit state $\ket{\varphi}$ can be expressed in the computational basis as $\ket{\varphi} = \sum_{x \in \bin^n} \alpha_x \ket{x}$ where $\alpha_x$ are complex amplitudes obeying $\sum_{x \in \bin^n} |\alpha_x|^2 =1$.

\paragraph*{Transformations.}
Evolutions of quantum systems are described by unitary transformations with $\identity_A$ being the identity transformation on register $A$. Given a joint quantum system over $\hilbert_A\otimes\hilbert_B$ and a transformation $U_A$ acting only on $\hilbert_A$, it is understood that $U_A\ket{\varphi_A}\ket{\varphi_B}$ refers to $(U_A\otimes\identity_B)\ket{\varphi_A}\ket{\varphi_B}$.

Information can be extracted from a quantum state $\ket{\varphi}$ by performing a positive-operator valued measurement (POVM) $M=\{M_i\}$ with positive semi-definite measurement operators $M_i$ that sum to the identity $\sum_i M_i = \identity$. Outcome $i$ is obtained with probability $p_i = \bra{\varphi} M_i\ket{\varphi}$.
A special case are projective measurements such as the measurement in the computational basis of the state $\ket{\varphi}=\sum_x \alpha_x \ket{x}$ which yields outcome $x$ with probability $|\alpha_x|^2$.  We can also do a partial measurement on some of the qubits.  The probability of the partial measurement resulting in a string $x$ is the same as if we measured the whole state, and ignored the rest of the qubits.  In this case, the resulting state will be the same as $|\phi\rangle$, except that all the strings inconsistent with $x$ are removed.  This new state will not have a norm of 1, so the actual superposition is obtained by dividing by the norm.  For example, if we measure the first $n$ bits of $|\phi\rangle=\sum_{x,y} \alpha_{x,y} |x,y\rangle$, we will obtain the measurement $x$ with probability $\sum_{y'}|\alpha_{x,y'}|^2$, and in this case the resulting state will be \eq{|x\rangle\sum_{y}\frac{\alpha_{x,y}}{\sqrt{\sum_{y'}|\alpha_{x,y'}|^2}}|y\rangle.}


Following \cite{BBCMW98}, we model a quantum attacker $\qadv$ with access to (possibly identical) oracles $O_1,O_2,\dots$ by a sequence of unitary transformations
$U_1, O_1, U_2,\dots, O_{T-1},U_T$ over $k=\text{poly}(n)$ qubits. Here, oracle $O_i:\bin^n\to\bin^m$ maps the first $n+m$ qubits from basis state $\ket{x}\ket{y}$ to basis state
$\ket{x}\ket{y\xor O_i(x)}$ for $x \in \bin^n$ and $y \in \bin^m$.
If we require the access to $O_i$ to be classical instead of quantum, the first $n$ bits of the state are measured before applying the unitary transformation corresponding to $O_i$.  Notice that any quantum-accessible oracle can also be used as a classical oracle.
Note that the algorithm $\qadv$ may also receive some input $\ket{\psi}$. 


To introduce asymptotics we assume that $\qadv$ is actually a sequence of such transformation sequences, indexed by parameter $n$, and that each transformation sequence is composed out of quantum systems for input, output, oracle calls, and work space (of sufficiently many qubits). To measure polynomial running time, we assume that each $U_i$ is approximated (to sufficient precision) by members of a set of universal gates (say, Hadamard, phase, CNOT and $\pi/8$; for sake of
concreteness \cite{ChaNie00}), where at most polynomially many gates are used. Furthermore, $T=T(n)$ is assumed to be polynomial, too.
Note that $T$ also bounds the number of oracle queries.

We define the Euclidean distance $\magn{|\phi\rangle-|\psi\rangle}$ between two states as the value $\parens{\sum_x |\alpha_x-\beta_x|^2}^{\frac{1}{2}}$ where $|\phi\rangle=\sum_x \alpha_x|x\rangle$ and $|\psi\rangle=\sum_x\beta_x|x\rangle$. 	

Define $q_r(|\phi_t\rangle)$ to be the magnitude squared of $r$ in the superposition of query $t$.  We call this the query probability of $r$ in query $t$.  If we sum over all $t$, we get the total query probability of $r$.
	
We will be using the following lemmas:
	\begin{lemma}[\cite{Bennett1997} Theorem 3.1] Let $|\varphi\rangle$ and $|\psi\rangle$ be quantum states with Euclidean distance at most $\epsilon$. Then,
performing the same measurement on $|\varphi\rangle$ and $|\psi\rangle$ yields distributions with statistical distance at most
$4\epsilon$.\end{lemma}

\begin{lemma}[\cite{Bennett1997} Theorem 3.3]\label{oraclechange} Let $A_Q$ be a quantum algorithm running in time $T$ with oracle access to $O$. Let $\epsilon > 0$ and let $S\subseteq\left[1,T\right]\times\{0,1\}^n$ be a set of time-string pairs such that $\sum_{(t,r)\in S}q_r(|\phi_t\rangle)\leq \epsilon$. If we modify $O$ into an oracle $O'$ which answers each query $r$ at time $t$ by providing the same string $R$ (which has been independently sampled at random), then the Euclidean distance between the final states of $A_Q$ when invoking $O$ and $O'$ is at most $\sqrt{T \epsilon}$.
\end{lemma}

\subsection{Quantum-Accessible Random Oracles}\label{quantum-access}

In the classical random oracle model~\cite{BR1} all algorithms used in
the system are given access to the same random oracle.  In the proof of security, the reduction algorithm answers the adversary's queries with
consistent random answers.

In the quantum settings, a quantum attacker issues a random oracle
query which is itself a superposition of exponentially many states.
The reduction algorithm must evaluate the random oracle at all points
in the superposition.  To ensure that random oracle queries are
answered consistently across queries,
it is convenient to assume that quantum-resistant
pseudorandom functions exist, and to implement this auxiliary random oracle with
such a PRF.

\begin{definition}[Pseudorandom Function] A quantum-accessible pseudorandom function is an efficiently computable function \prf\ 
where, for all efficient quantum algorithms $D$,
\eq{\magn{\Pr[D^{\PRF(k,\cdot)}(1^n)=1]-\Pr[D^{O(\cdot)}(1^n)=1]}<\epsilon}
where $\epsilon=\epsilon(n)$ is negligible in $n$, and where $O$ is a random oracle, the first probability is over the keys $k$ of length $n$, and the second probability is over all random oracles and the sampling of the result of $D$.
\end{definition}

We note that, following Watrous \cite{Wat09}, indistinguishability as above should still hold for any auxiliary quantum state $\sigma$ given as additional input to $D$ (akin to non-uniformity for classical algorithms). We do not include such auxiliary information in our definition in order to simplify.

We say that an oracle $O'$ is computationally indistinguishable from a random oracle if, for all polynomial time quantum algorithms with oracle access, the variational distance of the output distributions when the oracle is $O'$ and when the oracle is a truly random oracle $O$ is negligible.  Thus, simulating a random oracle with a quantum-accessible pseudorandom function is computationally indistinguishable from a true random oracle.

We remark that, instead of assuming that quantum-accessible PRFs exist, we can often carry out security reductions relative
to a random oracle. Consider, for example, a signature scheme (in the quantum-accessible random oracle model) which we prove to be
unforgeable for quantum adversaries, via a reduction to the one-wayness of a trapdoor permutation against quantum inverters.
We can then formally first claim that the scheme is unforgeable as long as inverting the trapdoor permutation is infeasible
even when having the additional power of a quantum-accessible random oracle; only in the next step we can then
conclude that this remains true in the standard model, if we assume that quantum-accessible pseudorandom functions exist and let the
inverter simulate the random oracle with such a PRF. We thus still get a potentially reasonable security claim even
if such PRFs do not exist.
This technique works whenever we can determine the success of the adversary
(as in case of inverting a one-way function).

\subsection{Hard Problems for Quantum Computers}

We will use the following general notion of a hard problem.

\begin{definition}[Problem]  A problem is a pair $P=(\myGame_P,\alpha_P)$ where  $\myGame_P$ specifies a game that a (possibly quantum) adversary plays with a classical challenger.  The game works as follows:
\begin{MyItemize}
	\item On input $1^n$, the challenger computes a value $x$, which it sends to the adversary as its input
	\item The adversary is then run on $x$, and is allowed to make classical queries to the challenger.
	\item The adversary then outputs a value $y$, which it sends to the challenger.
	\item The challenger then looks at $x$, $y$, and the classical queries made by the adversary, and outputs $1$ or $0$.
\end{MyItemize}
The value $\alpha_P$ is a real number between 0 (inclusive) and 1 (exclusive).  It may also be a function of $n$, but for this paper, we only need constant $\alpha_P$, specifically $\alpha_P$ is always $0$ or $\frac{1}{2}$.\end{definition}

We say that an adversary $A$ wins the game $\myGame_P$ if the challenger outputs $1$.  We define the advantage $\Adv_{A,P}$ of $A$ in problem $P$ as
\eq{Adv_{A,P}=\abs{\Pr[A\text{ wins in }\myGame_P]-\alpha_P}}
\begin{definition}[Hard Problem]
	A problem $P=(\myGame_P,\alpha_P)$ is hard for quantum computers if, for all polynomial time quantum adversaries $A$, $\Adv_{A,P}$ is negligible.
\end{definition}

\subsection{Cryptographic Primitives}

For this paper, we define the security of standard cryptographic primitives in terms of certain problems being hard for quantum computers.  We give a brief sketch here and refer to the \iffullversion appendix \fi  for supplementary details.

A trapdoor function $\mathcal{F}$ is secure if $\Inv(\mathcal{F})=(\myGame_{\INV}(\mathcal{F}),0)$ is a hard problem for quantum computers, where in $\myGame_{\INV}$, an adversary is given a random element $y$ and public key, and succeeds if it can output an inverse for $y$ relative to the public key.  A preimage sampleable trapdoor function, $\mathcal{F}$, is secure if $\Inv(\mathcal{F})$ as described above is hard, and if $\Col(\mathcal{F})=(\myGame_{\Col}(\mathcal{F}),0)$ is hard for quantum computers, where in $\myGame_{\Col}$, an adversary is given a public key, succeeds if it can output a collision relative to that public key.  A signature scheme $\mathcal{S}$ is secure if the game $\text{Sig-Forge}(\mathcal{S})=(\myGame_{\text{Sig}}(\mathcal{S}),0)$ is hard, where $\myGame_{\text{Sig}}$ is the standard existential unforgeability under a chosen message attack game.  Lastly, a private (resp.~public) key encryption scheme $\mathcal{E}$ is secure if $\text{Sym-CCA}(\mathcal{E})=(\myGame_{\text{Sym}}(\mathcal{E}),\frac{1}{2})$ (resp.~$\text{Asym-CCA}(\mathcal{E})=(\myGame_{\text{Asym}}(\mathcal{E}),\frac{1}{2})$), where $\myGame_{\text{Sym}}$ is the standard private key CCA attack game, and $\myGame_{\text{Asym}}$ is the standard public key attack game.


\section{Separation Result}

In this section, we discuss a two-party protocol that is provably secure in the random oracle model against both classical and quantum adversaries with classical access to the random oracle (and when using quantum-immune primitives). We then use the polynomial gap between the birthday attack and a collision finder based on Grover's algorithm to show that the protocol remains secure for certain hash functions when only classical adversaries are considered, but becomes insecure for any hash function if quantum adversaries are allowed. Analyzing the protocol in the stronger quantum random oracle model, where we grant the adversary quantum access to the random oracle, yields the same negative result.

\iffullversion

\subsection{Preliminaries}
We start this section by presenting the necessary definitions and assumptions for our construction. For sake of simplicity, we start with a quantum-immune identification scheme to derive our protocol; any other primitive or protocol can be used in a similar fashion.

\parag{Identification Schemes.}
An identification scheme $\ident$ consists of  three efficient algorithms \allowbreak $(\ikgen,\allowbreak\prov,\veri)$ where $\ikgen$ on input $1^n$ returns a key pair $(\sk,\pk)$. The joint execution of $\prov(\sk,\pk)$ and $\veri(\pk)$ then defines an interactive protocol between the prover $\prov$ and the verifier $\veri$. At the end of the protocol $\veri$ outputs a decision bit $b \in \bin$. We assume completeness in the sense that for any honest prover the verifier accepts the interaction with output $b=1$. Security of identification schemes is usually defined by considering an adversary $\adv$ that first interacts with the honest prover to obtain some information about the secret key. In a second stage, the adversary then plays the role of the prover and has to make a verifier accept the interaction. We say that an identification scheme is \emph{sound} if the adversary can convince the verifier with negligible probability only.

\parag{(Near-)Collision-Resistant Hash Functions.}

A hash function $\hash=(\hkgen,\heval)$ is a pair of efficient algorithms such that $\hkgen$ for input $1^n$ returns a key $k$ (which contains $1^n$), and $\heval$ for input $k$ and $M\in\bin^*$ deterministically outputs a digest $\heval(k,M)$. For a random oracle $H$ we use $k$ as a ``salt'' and consider the random function $H(k,\cdot)$. The hash function is called \emph{near-collision-resistant} if for any efficient algorithm $\adv$ the probability that for  $k\exec\hkgen(1^n)$, some constant $1 \leq \ell \leq n$ and $(M,M')\exec\adv(k,\ell)$ we have $M\neq M'$ but $\heval(k,M)|_{\ell}=\heval(k,M')|_{\ell}$, is negligible (as a function of $n$).
Here we denote by $x|_{\ell}$ the leading $\ell$ bits of the string $x$. Note that for $\ell=n$ the above definition yields the standard notion of collision-resistance.

In the classical setting, (near-)collision-resistance for any hash function is upper bounded by the \emph{birthday attack}. This generic attack states that for any hash function with $n$ bits output, an attacker can find a collision with probability roughly $1/2$ by probing $2^{n/2}$ distinct and random inputs. For random oracles this attack is optimal.

\parag{Grover's Algorithm and Quantum Collision Search.}
Grover's algorithm \cite{Gr96,Gro98} performs a search on an unstructured database with $N$ elements in time $O(\sqrt{N})$ while the best classical algorithm requires $O(N)$ steps. Roughly, this is achieved by using superpositions to examine all entries ``at the same time''. Brassard et al.~\cite{BraHoyTap98} showed that this speed-up can also be obtained for solving the collision problem for a hash function $H:\bin^* \rightarrow \bin^n$.  Therefore, one first selects a subset $K$ of the domain $\bin^*$ and then applies Grover's algorithm on an indicator function $f$ that tests for any input $M \in \bin^* \backslash K$ if there exists an $M' \in K$ such that  $H(M) = H(M')$ holds. By setting $|K| = \sqrt[3]{2^n}$, the algorithm finds a collision after $O(\sqrt[3]{2^n})$ evaluations of $H$ with probability at least $1/2$.

\parag{Computational and Timing Assumptions.}
To allow reasonable statements about the security of our protocol we need to formalize
assumptions concerning the computational power of the adversary and the time that elapses on quantum and classical computers.
We first demand that the speed-up one can gain by using a parallel machine with many processors, is bounded by a fixed term. This basically resembles the fact that in the real world there is only a concrete and finite amount of equipment available that can contribute to such a performance gain.

\newtheorem{assumption}{Assumption}

\begin{assumption}[Parallel Speed-Up]\label{assu:parallel}
Let $T(C)$ denote the time that is required to solve a problem $C$ on a classical computer, and $T_P(C)$ is the required time that elapses on a parallel system. Then, there exist a constant $\alpha \geq 1$, such that for any problem $C$ it holds that $T_P(C) \geq T(C) / \alpha$.
\end{assumption}

We also introduce two assumptions regarding the time that is needed to evaluate a hash function or to send a message between two parties. Note that both assumptions are merely for the sake of convenience, as one could patch the idea by relating the timings more rigorously. The first assumption states that the time that is required to evaluate a hash function $H$ is independent of the input and the computational environment.
\begin{assumption}[Unit Time]
For any hash function $H$ and any  input message $M$ (resp. $M_Q$ for quantum-state inputs) the evaluation of $H(M)$ requires a constant time $T(H(M)) = T_P(H(M))= T_Q(H(M_Q))$ (where $T_Q$ denotes the time that elapses on a quantum computer).
\end{assumption}

Furthermore, we do not charge any extra time for sending and receiving messages, or for any computation other than evaluating a hash function (e.g., maintaining lists of values).
\begin{assumption}[Zero Time]\label{ass:zero}
Any computation or action that does not require the evaluation of a hash function, costs zero time.
\end{assumption}
The latter assumption implicitly states that the computational overhead that quantum algorithms may create to obtain a speed-up is negligible when compared to the costs of a hash evaluation. This might be too optimistic in the near future, as indicated by Bernstein~\cite{Ber09}. That is, Bernstein discussed that the overall costs of a quantum computation can be higher than of massive parallel computation. However, as our work addresses conceptional issues that arise when \emph{efficient} quantum computers exist, this assumption is somewhat inherent in our scenario.

\fi

%

\subsection{Construction}

\iffullversion

We now present our identification scheme between a prover $\prov$ and
a verifier $\veri$ (see Figure~\ref{fig:protocol}) The main idea is to augment a secure identification scheme $\ident$ by a collision-finding stage for some hash function $\hash$. In this first stage, the verifier checks if the prover is able to produce collisions on a hash function in a particular time. More precisely, the verifier starts for timekeeping to evaluate the hash function $\heval(k,\cdot)$ on the messages $\binrep{c}$ for $c=1,2, \dots, \left\lceil \sqrt[3]{2^\ell}\right\rceil$ for a key $k$ chosen by the verifier and where $\binrep{c}$ stands for the binary representation of $c$ with $\log{\left\lceil \sqrt[3]{2^\ell}\right\rceil}$ bits. The prover has now to respond with a near-collision $M\neq M'$ such that $\heval(k,M) = \heval(k,M')$ holds for the first $\ell$ bits. One round of the collision-stage ends if the verifier either receives such a collision or finishes its $\sqrt[3]{2^\ell}$ hash evaluations. The verifier and the receiver then repeat such a round $r=\text{poly}(n)$ times, sending a fresh key $k$ in each round.

Subsequently, both parties run the standard identification scheme. At the end, the verifier accepts if the prover was able to find enough collisions in the first stage or identifies correctly in the second stage.
Thus, as long as the prover is not able to produce collisions in the required time, the protocol mainly resembles the $\ident$ protocol.

\fi

\begin{figure}[h!]
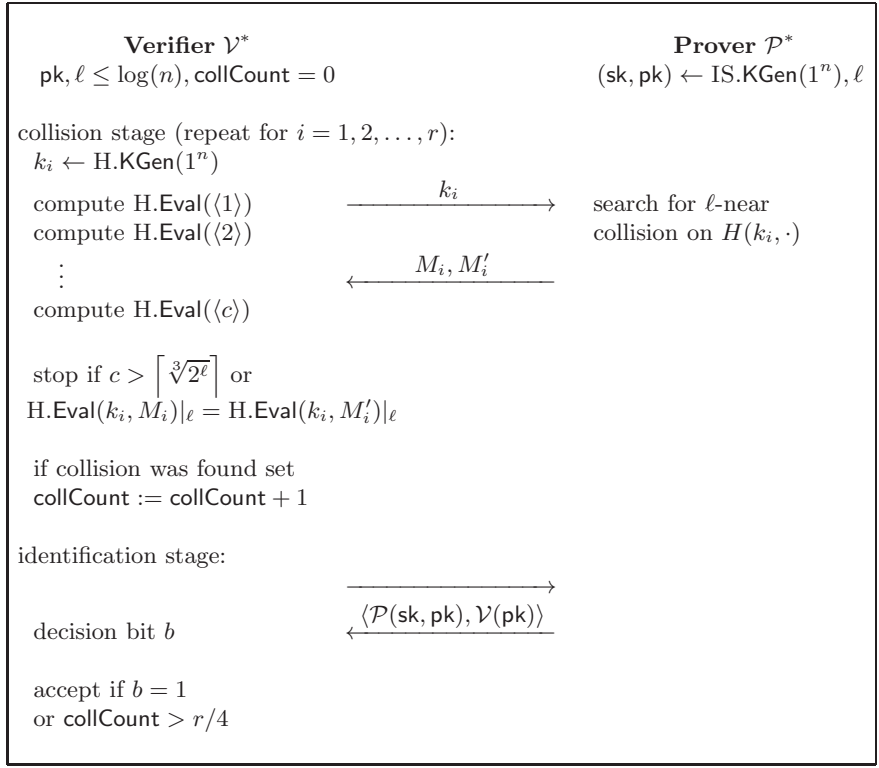

\begin{center}
{
\begin{tabular}{|@{\hspace*{0.75em}}ll@{\hspace*{1.5em}}c@{\hspace*{1.5em}}l@{\hspace*{0.75em}}|}
\hline
& & & \\
&\multicolumn{1}{c}{\textbf{Verifier} $\veri^*$} & & \multicolumn{1}{c|}{\textbf{Prover} $\prov^*$} \\
& \multicolumn{1}{c}{$\pk,  \ell \leq \log(n),\ccount=0  $}  &  & \multicolumn{1}{c|}{$(\sk,\pk) \exec \ikgen(1^n), \ell$ \qquad} \\
& & & \\
\multicolumn{3}{|l}{\;collision stage (repeat for $i=1,2,\dots,r$):} &  \\ 
& $k_i \exec \hkgen(1^n)$ &  & \\
& compute $\heval(\binrep{1})$  & $\xrightarrow{\makebox[8em]{$k_i$}}$ & search for $\ell$-near   \\
& compute $\heval(\binrep{2})$ &  & collision on $H(k_i,\cdot)$   \\
&\quad\vdots &$\xleftarrow{\makebox[8em]{$M_i,M_i'$}}$ &    \\
& compute $\heval(\binrep{c})$ &  & \\
&&& \\
& stop if $c > \left\lceil \sqrt[3]{2^\ell}\right\rceil$ or  &&\\
\multicolumn{3}{|@{\hspace*{0.75em}}l}{$\heval(k_i,M_i)|_{\ell} = \heval(k_i,M_i')|_{\ell}$}& \\
& & & \\
& if collision was found set&& \\
& $\ccount:=\ccount +1$ &&\\
&& & \\
\multicolumn{2}{|l}{\;identification stage:} & & \\
& &$\xrightarrow{\makebox[8em]{}}$ &    \\
&decision bit $b$ &$\xleftarrow{\makebox[8em]{$\left\langle \prov(\sk,\pk), \veri(\pk)\right\rangle$}}$ &    \\
&&& \\

&accept if $b=1$&& \\
&or $\ccount > r/4$ &&\\


&& & \\ \hline
\end{tabular}
}
\end{center}
\caption{\small The $\ident^*$-Identification Protocol}
\label{fig:protocol}
\end{figure}

Completeness of the $\ident^*$ protocol follows easily from the completeness of the underlying $\ident$ scheme.

\parag{Security against Classical and Quantum Adversaries.}
To prove security of our protocol, we need to show that an adversary $\adv$ after interacting with an honest prover $\prov^*$, can subsequently not impersonate $\prov^*$ such that $\veri^*$ will accept the identification.
Let $\ell$ be such that $\ell > 6 \log(\alpha)$ where $\alpha$ is the constant reflecting the bounded speed-up in parallel computing from Assumption (1). By assuming that $\ident=(\ikgen,\allowbreak\prov,\veri)$ is a quantum-immune identification scheme, we can show that $\ident^*$ is secure in the standard random oracle model against classical and quantum adversaries.

The main idea is that for the standard random oracle model, the ability of finding collisions is bounded by the birthday attack. Due to the constraint of granting only time $O(\sqrt[3]{2^\ell})$ for the collision search and setting $\ell > 6 \log(\alpha)$, even an adversary with quantum or parallel power is not able to make at least $\sqrt{2^\ell}$ random oracle queries. Thus, $\adv$ has only negligible probability to respond in more than $1/4$ of $r$ rounds with a collision.

When considering only classical adversaries, we can also securely instantiate the random oracle by a hash function $H$ that provides near-collision-resistance close to the birthday bound. Note that this property is particularly required from the SHA-3 candidates~\cite{sha3}.

However, for adversaries $\qadv$ with quantum power, such an instantiation is not possible for \emph{any} hash function. This stems from the fact that $\qadv$ can locally evaluate a hash function on quantum states which in turns allows to apply Grover's search algorithm. Then an adversary will find a collision in time
$\sqrt[3]{2^\ell}$ with probability at least $1/2$, and thus will be able to provide $r/4$ collisions with noticeable probability. The same result holds in the quantum-accessible random oracle model, since Grover's algorithm only requires (quantum) black-box access to the hash function.

Formal proofs of all statements are given in \iffullversion Appendix~\ref{app:negative}. \fi


\section{Signature Schemes in the Quantum-Accessible Random Oracle Model}\label{sig}

We now turn to proving security in the quantum-accessible random oracle model.  We present general conditions for when a proof of security in the classical random oracle model implies security in the quantum-accessible random oracle model.  The result in this section applies to signatures whose classical proof of security is a \emph{\os\ reduction} as defined next. Roughly speaking, history-freeness means that the classical proof of security simulates the random oracle and signature oracle in a \os\ fashion.  That is, its responses to queries do not depend on responses to previous queries or the query number.   We then show that a number of classical signature schemes have a \os\ reduction thereby proving their security in the quantum-accessible random oracle model.

\begin{definition}[\OS\ Reduction]\label{def:signatureReduction}
	A random oracle model signature scheme $\mathcal{S}=(G,S^O,V^O)$ has a \os\ reduction from a hard problem $P=(\myGame_P,0)$ if there is a proof of security that uses a classical PPT adversary $A$ for $\mathcal{S}$ to construct a classical PPT algorithm $B$ for problem $P$ such that:
	\begin{MyItemize}
		\item Algorithm $B$ for $P$ contains four explicit classical algorithms: $\START$, $\RAND^{O_c}$, $\SIGN^{O_c}$, and $\FINISH^{O_c}$.  The latter three algorithms have access to a shared classical random oracle $O_c$.  These algorithms, except for $\RAND^{O_c}$, may also make queries to the challenger for problem $P$.   The algorithms are used as follows:
			\begin{MyItemize}
				\item[(1)] Given an instance $x$ for problem $P$ as input, algorithm $B$ first runs $\START(x)$ to obtain $(\pk,z)$ where $\pk$ is a signature public key and $z$ is private state to be used by $B$.   Algorithm $B$ sends $\pk$ to $A$ and plays the role of challenger to $A$.
				\item[(2)] When $A$ makes a classical random oracle query to $O(r)$, algorithm $B$ responds with $\RAND^{O_c}(r,z)$.   Note that $\RAND$ is given the current query as input, but is unaware of previous queries and responses.
				\item[(3)] When $A$ makes a classical signature query $S(\sk,m)$, algorithm $B$ responds with $\SIGN^{O_c}(m,z)$.
				\item[(4)] When $A$ outputs a signature forgery candidate $(m,\sigma)$, algorithm $B$ outputs $\FINISH^{O_c}(m,\sigma,z)$.
			\end{MyItemize}
		\item There is an efficiently computable function $\INSTANCE(\pk)$ which produces an instance $x$ of problem $P$ such that $\START(x)=(\pk,z)$ for some~$z$.  Consider the process of first generating $(\sk,\pk)$ from $G(1^n)$, and then computing $x=\INSTANCE(\pk)$.  The distribution of $x$ generated in this way is negligibly close to the distribution of $x$ generated in $\myGame_P$.
		\item For fixed $z$, consider the classical random oracle $O(r)=\RAND^{O_c}(r,z)$.  Define a quantum oracle $O_\textrm{\rm quant}$, which transforms a basis element $\ket{x,y}$ into $\ket{x,y\oplus O(x)}$.  We require that $O_\textrm{\rm quant}$ is quantum computationally indistinguishable from a random oracle. 
		\item $\SIGN^{O_c}$ either aborts (and hence $B$ aborts) or it generates a valid signature relative to the oracle $O(r)=\RAND^{O_c}(r,z)$ with a distribution negligibly close to the correct signing algorithm.  The probability that none of the signature queries abort is non-negligible.
		\item If $(m,\sigma)$ is a valid signature forgery relative to the public key $\pk$ and oracle $O(r)=\RAND^{O_c}(r,z)$ then the output of $B$ (i.e. $\FINISH^{O_c}(m,\sigma,z)$) causes the challenger for problem $P$ to output $1$ with non-negligible probability.\qed
	\end{MyItemize} \end{definition}


\noindent We now show that \os\ reductions imply security in the quantum setting.
	
	\begin{theorem}\label{metasignature}Let $\mathcal{S}=(G,S,V)$ be a signature scheme.  Suppose that there is a \os\ reduction that uses a classical PPT adversary $A$ for $\mathcal{S}$ to construct a PPT algorithm $B$ for a problem $P$.  Further, assume that $P$ is hard for polynomial-time quantum computers, and that quantum-accessible pseudorandom functions exist.  Then $\mathcal{S}$ is secure in the quantum-accessible random oracle model.\end{theorem}
		
	\begin{proof}
	The \os\ reduction includes five (classical) algorithms $\START$, $\RAND$, $\SIGN$, $\FINISH$, and $\INSTANCE$, as in Definition~\ref{def:signatureReduction}.
	We prove the quantum security of $\mathcal{S}$ using a sequence of games, where the first game is the standard quantum signature game with respect to $\mathcal{S}$.
	
\medskip\noindent{\bf Game 0.} Define $\myGame_0$ as the game a quantum adversary $A_Q$ plays for problem $\text{Sig-Forge}(\mathcal{S})$.  Assume towards contradiction that $A_Q$ has a non-negligible advantage.

\medskip\noindent{\bf Game 1.} Define $\myGame_1$ as the following modification to $\myGame_0$: after the challenger generates $(\sk,\pk)$, it computes $x \gets\INSTANCE(\pk)$ as well as $(\pk,z) \gets \START(x)$.  Further, instead of answering $A_Q$'s quantum random oracle queries with a truly random oracle, the challenger simulates for $A_Q$ a quantum-accessible random oracle $O_\textrm{\rm quant}$ as an oracle that maps a basis element $\ket{x,y}$ into the element $\ket{x,y\oplus \RAND^{O_c}(x,z)}$, where $O_c$ is a truly (classical) random oracle.
The \os\ guarantee on $\RAND$ ensures that $O_\textrm{\rm quant}$ is computationally indistinguishable from random for quantum adversaries.  Therefore, the success probability of $A_Q$ in $\myGame_1$ is negligibly close to its success probability in $\myGame_0$, and hence is non-negligible.

\medskip\noindent{\bf Game 2.} Modify the challenger from $\myGame_1$ as follows: instead of generating $(\sk,\pk)$ and computing $x=\INSTANCE(\pk)$, start off by running the challenger for problem $P$.  When that challenger sends $x$, then start the challenger from $\myGame_1$ using this $x$.  Also, when $A_Q$ asks for a signature on $m$, answer with $\SIGN^{O_c}(m,z)$.  First, since $\INSTANCE$ is part of a \os\ reduction, this change in how we compute $x$ only negligibly affects the distribution of $x$, and hence the behavior of $A_Q$.  Second, as long as all signing algorithms succeed, changing how we answer signing queries only negligibly affects the behavior of $A_Q$.  Thus, the probability that $A_Q$ succeeds is the product of the following two probabilities: \begin{MyItemize}\item The probability that all of the signing queries are answered without aborting. \item The probability that $A_Q$ produces a valid forgery given that the signing queries were answered successfully.\end{MyItemize} The first probability is non-negligible by assumption, and the second is negligibly close to the success probability of $A_Q$ in $\myGame_1$, which is also non-negligible.  This means that the success probability of $A_Q$ in $\myGame_2$ is non-negligible.

\medskip\noindent{\bf Game 3.} Define $\myGame_3$ as in $\myGame_2$, except for two modifications to the challenger:  First, it generates a key $k$ for the quantum-accessible $\PRF$.  Then, to answer a random oracle query $O_\textrm{\rm quant}(\ket{\phi})$, the challenger applies the unitary transformation that takes a basis element $\ket{x,y}$ into $\ket{x,y\oplus\PRF(k,x)}$.  If the success probability in $\myGame_3$ was non-negligibly different from that of $\myGame_2$, we could construct a distinguisher for $\PRF$ which plays both the role of $A_Q$ and the challenger.  Hence, the success probability in $\myGame_3$ is negligibly close to that of $\myGame_2$, and hence is also non-negligible.

Given a quantum adversary that has non-negligible advantage in Game~3 we construct a
quantum algorithm $B_Q$ that breaks problem $P$.  When $B_Q$ receives instance $x$ from the challenger for problem $P$, it computes $(\pk,z) \gets \START(x)$ and generates a key $k$ for $\PRF$.  Then, it simulates $A_Q$ on $\pk$.  $B_Q$ answers random oracle queries using a quantum-accessible function built from $\RAND^{\PRF(k,\cdot)}(\cdot,z)$ as in Game~1.  It answers signing queries using $\SIGN^{\PRF(k,\cdot)}(\cdot,z)$.   Then, when $A_Q$ outputs a forgery candidate $(m,\sigma)$, $B_Q$ computes $\FINISH^{\PRF(k,\cdot)}(m,\sigma,z)$, and returns the result to the challenger for problem $P$.

Observe that the behavior of $A_Q$ in $\myGame_3$ is identical to that as a subroutine of $B_Q$.  Hence, $A_Q$ as a subroutine of $B_Q$ will output a valid forgery $(m,\sigma)$ with non-negligible probability.  If $(m,\sigma)$ is a valid forgery, then since $\FINISH$ is part of a \os\ reduction, $\FINISH^{\PRF(k,\cdot)}(m,\sigma,z)$ will cause the challenger for problem $P$ to accept with non-negligible probability.  Thus, the probability that $P$ accepts is also non-negligible, contradicting our assumption that $P$ is hard for quantum computers.

Hence we have shown that any polynomial quantum algorithm has negligible advantage against problem $\text{Sig-Forge}(\mathcal{S})$ which completes the proof. \qed
\end{proof}

We note that, in every step of the algorithm, the adversary $A_Q$ remains in a pure state.  This is because, in each game, $A_Q$'s state is initially pure (since it is classical), and every step of the game either involves a unitary transformation, a partial measurement, or classical communication.  In all three cases, if the state is pure before, it is also pure after.

We also note that we could have stopped at $\myGame_2$ and assumed that the cryptographic problem $P$ is hard relative to a (quantum-accessible)
random oracle. Assuming the existence of quantum-accessible pseudorandom functions allows us to draw the same conclusion in the standard
(i.e., non-relativized) model at the expense of an extra assumption.

\subsection{Secure Signatures From Preimage Sampleable Trapdoor Functions (PSF)}\label{sec:sigsfrompsf}

We now use Theorem \ref{metasignature} to prove the security of the Full Domain Hash signature scheme when instantiated with a preimage sampleable trapdoor function (PSF), such as the one proposed in~\cite{Gentry2008}.   Loosely speaking, a PSF $\mathcal{F}$ is a tuple of PPT algorithms $(G,\text{Sample},f,f^{-1})$ where $G(\cdot)$ generates a key pair $(\pk,\sk)$,\ \ $f(\pk,\cdot)$ defines an efficiently computable function,\  $f^{-1}(\sk,y)$ samples from the set of pre-images of $y$, and $\text{Sample}(\pk)$ samples $x$ from the domain of $f(\pk,\cdot)$ such that $f(\pk,x)$ is statistically close to uniform in the range of $f(\pk,\cdot)$.   The PSF of~\cite{Gentry2008}  is not only one-way, but is also collision resistant.

\medskip
Recall that the full domain hash (FDH) signature scheme~\cite{BR1} is defined as follows:
\begin{definition}[Full Domain Hash] \label{def:FDH}
Let $\mathcal{F}=(G_0,f,f^{-1})$ be a trapdoor permutation, and $O$ a hash function whose range is the same as the range of $f$.  The full domain hash signature scheme is $\mathcal{S}=(G,T,V)$ where:
	\begin{MyItemize}
		\item $G=G_0$
		\item $S^O(\sk,m)=f^{-1}(\sk,O(m))$
		\item $V^O(\pk,m,\sigma)=\begin{cases}1&\text{if }O(m)=f(\pk,\sigma)\\0&\text{otherwise}\end{cases}$
	\end{MyItemize}\end{definition}


Gentry et al.~\cite{Gentry2008} show that the FDH signature scheme can be instantiated with a PSF $\mathcal{F}=(G,\text{Sample},f,f^{-1})$ instead of a trapdoor permutation.  Call the resulting system FDH-PSF.  They prove that FDH-PSF is secure against classical adversaries, provided that the pre-image sampling algorithm used during signing is derandomized (e.g. by using a classical PRF to generate its random bits).  Their reduction is not quite \os, but we show that it can be made \os.

Consider the following reduction from a classical adversary~$A$ for the FDH-PSF scheme $\mathcal{S}$ to a classical collision finder $B$ for $\mathcal{F}$:
\begin{MyItemize}
	\item On input $\pk$, $B$ computes $\START(\pk) \deq (\pk,\pk)$, and simulates $A$ on $\pk$.
  \item When $A$ queries $O(r)$, $B$ responds with \\ \mbox{} \hspace{4cm}
  		$\RAND^{O_c}(r,\pk) \deq f(\pk,\text{Sample}(1^n;O_c(r)))$.
  \item When $A$ queries $S(\sk,m)$, $B$ responds with \\ \mbox{} \hspace{4cm}
  		$\SIGN^{O_c}(m,\pk) \deq Sample(1^n;O_c(m))$.
  \item When $A$ outputs $(m,\sigma)$, $B$ outputs \\ \mbox{} \hspace{4cm}
  		$\FINISH^{O_c}(m,\sigma,\pk) \deq \big(Sample\big(1^n;O_c(m)\big),\sigma\big)$.
\end{MyItemize}
In addition, we define $\INSTANCE(\pk) \deq \pk$.  Algorithms $\INSTANCE$ and $\START$ trivially satisfy the requirements of history-freeness (Definition~\ref{def:signatureReduction}).   Before showing that the above reduction is in \os\ form, we need the following technical lemma whose proof is given in the \iffullversion appendix \fi .

	\begin{lemma}\label{oraclerand}
		Say $A$ is a quantum algorithm that makes $q$ quantum oracle queries.  Suppose further that we draw the oracle $O$ from two distributions.  The first is the random oracle distribution.  The second is the distribution of oracles where the value of the oracle at each input $x$ is identically and independently distributed by some distribution $D$ whose variational distance is within $\epsilon$ from uniform.  Then the variational distance between the distributions of outputs of $A$ with each oracle is at most $4 q^2 \sqrt{\epsilon}$.
	\end{lemma}
	
	\smallskip\noindent {\bf Proof Sketch.}
	We show that there is a way of moving from $O$ to $O_D$ such that the oracle is only changed on inputs in a set $K$ where the sum of the amplitudes squared of all $k\in K$, over all queries made by $A$, is small.  Thus, we can use Lemma~\ref{oraclechange} to show that the expected behavior of any algorithm making polynomially many quantum queries to $O$ is only changed by a small amount.   \qed
	
	Lemma~\ref{oraclerand} shows that we can replace a truly random oracle $O$ with an oracle $O_D$ distributed according to distribution $D$ without impacting $A$, provided $D$ is close to uniform.  Note, however, that while this change only affects the output of $A$ negligibly, the effects are larger than in the classical setting.  If $A$ only made classical queries to $O$, a simple hybrid argument shows that changing to $O_D$ affects the distribution of the output of $A$ by at most $q \epsilon$, as opposed to $4 q^2 \sqrt{\epsilon}$ in the quantum case.  Thus, quantum security reductions that use Lemma~$\ref{oraclerand}$ will not be as tight as their classical counterparts.	 
	
	\medskip\noindent	We now show that the reduction above is \os.
	
\begin{theorem} \label{thm:statsig}
The reduction above applied to FDH-PSF is \os.
\end{theorem}
\begin{proof}
	The definition of a PSF implies that the distribution of $f(\pk,\text{Sample}(1^n))$ is within $\epsilon_{\text{sample}}$ of uniform, for some negligible $\epsilon_{\text{sample}}$.   Now, since $O(r)=\RAND^{O_c}(r,\pk)=f(\pk,\text{Sample}(1^n;O_c(r)))$ and $O_c$ is a true random oracle, the quantity $O(r)$ is distributed independently according to a distribution that is $\epsilon_{\text{sample}}$ away from uniform.  Define a quantum oracle $O_\textrm{\rm quant}$ which transforms the basis state $\ket{x,y}$ into $\ket{x,y\oplus O(x)}$.  Using Lemma~\ref{oraclerand}, for any algorithm $B$ making $q$ random oracle queries, the variational distance between the probability distributions of the outputs of $B$ using a truly random oracle and the ``not-quite'' random oracle $O_\textrm{\rm quant}$ is at most $4q^2\sqrt{\epsilon_{\text{sample}}}$, which is still negligible.  Hence, $O_\textrm{\rm quant}$ is computationally indistinguishable from random.
	
	Gentry et al.~\cite{Gentry2008} also show that $\SIGN^{O_c}(m,\pk)$ is consistent with $\RAND^{O_c}(\cdot,\pk)$ for all queries, and that if $A$ outputs a valid forgery $(m,\sigma)$, $\FINISH^{O_c}(m,\sigma,\pk)$ produces a collision for $\mathcal{F}$ with probability $1-2^{-E}$, where $E$ is the minimum over all $y$ in the range of $f(\pk,\cdot)$ of the min-entropy of the distribution on $\sigma$ given $f(\pk,\sigma)=y$.  The PSF of Gentry et al.~\cite{Gentry2008} has super-logarithmic min-entropy, so $1-2^{-E}$ is negligibly close to 1, though any constant non-zero min-entropy will suffice to make the quantity a non-negligible fraction of 1.\qed
\end{proof}

	We note that the security proof of Gentry et al.~\cite{Gentry2008} is a tight reduction in the following sense: if the advantage of an adversary $A$ for $\mathcal{S}$ is $\epsilon$, the reduction gives a collision finding adversary $B$ for $\mathcal{F}$ with advantage negligibly close to $\epsilon$, provided that the lower bound over $y$ in the range of $f(\pk,\cdot)$ of the min-entropy of $\sigma$ given $f(\pk,\sigma)=y$ is super-logarithmic.  If the PSF has a min-entropy of 1, the advantage of $B$ is still $\epsilon/2$.

\noindent
The following corollary, which is the main result of this section, follows from Theorems~{(\ref{metasignature})} and~{(\ref{thm:statsig})}.

\begin{corollary}\label{cor:psf}If quantum-accessible pseudorandom functions exist, and $\mathcal{F}$ is a secure PSF against quantum adversaries, then the FDH-PSF signature scheme is secure in the quantum-accessible random oracle model.
\end{corollary}

\subsection{Secure Signatures from Claw-Free Permutations}

In this section, we show how to use claw-free permutations to construct three signature schemes that have \os\ reductions and are therefore secure in the quantum-accessible random oracle model.  The first is the standard FDH from Definition~\ref{def:FDH}, but when the underlying permutation is a claw-free permutation.  We adapt the proof of Coron~\cite{Coron2000} to give a \os\ reduction.  The second is the Katz and Wang~\cite{Katz2003} signature scheme, and we also modify their proof to get a \os\ reduction.  Lastly, following Gentry et al.~\cite{Gentry2008}, we note that claw-free permutations give rise to a pre-image sampleable trapdoor function (PSF), which can then be used in FDH to get a secure signature scheme as in Section~\ref{sec:sigsfrompsf}.  The Katz-Wang and FDH-PSF schemes from claw-free permutations give a tight reduction, whereas the Coron-based proof loses a factor of $q_s$ in the security reduction, where $q_s$ is the number of signing queries.

Recall that a claw-free pair of permutations~\cite{Goldwasser1988} is a pair of trapdoor permutations $(\mathcal{F}_1,\mathcal{F}_2)$, where $\mathcal{F}_i=(G_i,f_i,f^{-1}_i)$, with the following properties:
\begin{MyItemize}
	\item $G_1=G_2$.  Define $G=G_1=G_2$.
	\item For any key $\pk$, $f_1(\pk,\cdot)$ and $f_2(\pk,\cdot)$ have the same domain and range.
	\item Given only $\pk$, the probability that any PPT adversary can find a pair $(x_1,x_2)$ such that $f_1(\pk,x_1)=f_2(\pk,x_2)$ is negligible.  Such a pair is called a claw.
\end{MyItemize}

Dodis and Reyzin~\cite{Dodis2003} note that claw-free permutations are a generalization of trapdoor permutations with a random self-reduction.  A random self-reduction is a way of taking a worst-case instance~$x$ of a problem, and converting it into a random instance $y$ of the same problem, such that a solution to $y$ gives a solution to $x$.   Dodis and Reyzin~\cite{Dodis2003} show that any trapdoor permutation with a random self reduction (e.g. RSA) gives a claw-free pair of permutations.

We note that currently there are no candidate pairs of claw-free permutations that are secure against quantum adversaries, but this may change in time.

\subsubsection*{FDH Signatures from Claw-Free Permutations}

Coron~\cite{Coron2000} shows that the Full Domain Hash signature scheme, when instantiated with the RSA trapdoor permutation, has a tighter security reduction than the general Full Domain Hash scheme, in the classical world.  That is, Coron's reduction loses a factor of approximately $q_s$, the number of signing queries, as apposed to $q_h$, the number of hash queries.  Of course, the RSA trapdoor permutation is not secure against quantum adversaries, but his reduction can be applied to any claw-free permutation and is equivalent to a \os\ reduction with similar tightness.

To construct a FDH signature scheme from a pair of claw-free permutations $(\mathcal{F}_1,\mathcal{F}_2)$, we simply instantiate FDH with $\mathcal{F}_1$, and ignore the second permutation $\mathcal{F}_2$, to yield the following signature scheme
\begin{MyItemize}
 \item $G$ is the generator for the pair of claw-free permutations.
 \item $S^O(\sk,m) = f_1^{-1}(\sk,O(m))$
 \item $V^O(\pk,m,\sigma)=1$ if and only if $f_1(\pk,\sigma)=O(m)$.
\end{MyItemize}

We now present a \os\ reduction for this scheme.  The random oracle for this reduction, $O_c(r)$, returns a random pair $(a,b)$, where $a$ is a random element from the domain of $\mathcal{F}_1$ and $\mathcal{F}_2$, and $b$ is a random element from $\{1,...,p\}$ for some $p$ to be chosen later.

We construct \os\ reduction from a classical adversary $A$ for $\mathcal{S}$ to a classical adversary $B$ for $(\mathcal{F}_1,\mathcal{F}_2)$.  Algorithm~$B$,  on input $\pk$, works as follows:
  \begin{MyItemize}
  \item Compute $\START(\pk,y)=(\pk,\pk)$, and simulate $A$ on $\pk$.  Notice that $z=\pk$ is the state saved by $B$.
	\item When $A$ queries $O(r)$, compute $\RAND^{O_c}(r,\pk)$.  For each string $r$, $\RAND$ works as follows: compute $(a,b) \gets O_c(r)$.  If $b=1$, return $f_2(\pk,a)$.  Otherwise, return $f_1(\pk,a)$
	\item When $A$ queries $S(\sk,m)$, compute $\SIGN^{O_c}(m,\pk)$.  $\SIGN$ works as follows: compute $(a,b) \gets O_c(m)$ and return $a$ if $b\neq1$.  Otherwise, fail.
	\item When $A$ returns $(m,\sigma)$, compute $\FINISH^{O_c}(m,\sigma,\pk)$.  $\FINISH$ works as follows: compute $(a,b) \gets O_c(m)$ and output $(\sigma,a)$.
  \end{MyItemize}
  In addition, we have $\INSTANCE(\pk)=\pk$ and $\START(\INSTANCE(\pk))=(\pk,\pk)$, so $\INSTANCE$ and $\START$ satisfy the required properties.

\begin{theorem}The reduction above is in \os\ form.\end{theorem}
\begin{proof}
  $\RAND^{O_c}(r,\pk)$ is completely random and independently distributed, as $f_1(\pk,a)$ and $f_2(\pk,a)$ are both random ($f_b(\pk,\cdot)$ is a permutation and $a$ is truly random).  As long as $b\neq 1$, where $(a,b)=O_c(m)$, $\SIGN^{O_c}(m,\pk)$ will be consistent with $\RAND$.  This is because $V^{\RAND^{O_c}(\cdot,\pk)}(\pk,m,\SIGN^{O_c}(m,\pk))$ outputs 1 if $\RAND^{O_c}(m,\pk)=f_1(\pk,\SIGN^{O_c}(m,\pk))$.  But $\RAND^{O_c}(m,\pk)=f_1(\pk,a)$ (since $b\neq1$), and $\SIGN^{O_c}(m,\pk))=a$.  Thus, the equality holds.  The probability over all signature queries of no failure is $(1-1/p)^{q_{\SIGN}}$.  If we chose $p=q_{\SIGN}$, this quantity is at least $e^{-1}-o(1)$, which is non-negligible.

  Suppose $A$ returns a valid forgery $(m,\sigma)$, meaning $A$ never asked for a forgery on $m$ and $f_1(\sk,\sigma)=\RAND^{O_c}(m,\pk)$.  If $b=1$ (where $(a,b)=O_c(m)$), then we have $f_1(\sk,\sigma)\allowbreak =\RAND^{O_c}(m,\pk)=f_2(\pk,a)$, meaning that $(\sigma,a)$ is a claw.  Since $A$ never asked for a signature on $m$, there is no way $A$ could have figured out $a$, so the case where $b=1$ and $a$ is the preimage of $O(m)$ under $f_2$, and the case where $b\neq 1$ and $a$ is the preimage of $O(m)$ under $f_1$ are indistinguishable.   Thus, $b=1$ with probability $1/p$.  Thus, $B$ converts a valid signature into a claw with non-negligible probability.\qed\end{proof}
  \begin{corollary} If quantum-accessible pseudorandom functions exists, and $(\mathcal{F}_1,\mathcal{F}_2)$ is a pair claw-free trapdoor permutations, then the FDH scheme instantiated with $\mathcal{F}_1$ is secure against quantum adversaries.\end{corollary}

Note that in this reduction, our simulated random oracle is truly random, so we do not need to rely on Lemma~\ref{oraclerand}.  Hence, the tightness of the reduction will be the same as the classical setting.  Namely, if the quantum adversary $A$ has advantage $\epsilon$ when making $q_\SIGN$ signature queries, $B$ will have advantage approximately $\epsilon/ q_\SIGN$.

\subsubsection*{The Katz-Wang Signature Scheme}

In this section, we consider a variant of FDH due to Katz and Wang~\cite{Katz2003}.  This scheme admits an almost tight security reduction in the classical world.  That is, if an adversary has advantage $\epsilon$, the reduction gives a claw finder with advantage $\epsilon/2$.  Their proof of security is not in \os\ form, but it can be modified so that it is in \os\ form.  Given a pair of trapdoor permutation $(\mathcal{F}_1,\mathcal{F}_2)$, the construction is as follows:
\begin{MyItemize}
	\item $G$ is the key generator for $\mathcal{F}$.
	\item $S^O(\sk,m)=f_1^{-1}(\sk,O(b,m))$ for a random bit $b$.
	\item $V^O(\pk,m,\sigma)$ is 1 if either $f_1(\pk,\sigma)=O(0,m)$ or $f_1(\pk,\sigma)=O(1,m)$
\end{MyItemize}

We construct a \os\ reduction from an adversary $A$ for $\mathcal{S}$ to an adversary $B$ for $(\mathcal{F}_1,\mathcal{F}_2)$.  The random oracle for this reduction, $O_c(r)$, generates a random pair $(a,b)$, where $a$ is a random element from the domain of $\mathcal{F}_1$ and $\mathcal{F}_2$, and $b$ is a random bit.  On input $\pk$, $B$ works as follows:
  \begin{MyItemize}
  \item Compute $\START(\pk,y)=(\pk,\pk)$, and simulate $A$ on $\pk$.  Notice that $z=\pk$ is the state saved by $B$.
	\item When $A$ queries $O(b,r)$, compute $\RAND^{O_c}(b,r,\pk)$.  For each string $(b,r)$, $\RAND$ works as follows: compute $(a,b')=O_c(r)$.  If $b=b'$, return $f_1(\pk,a)$.  Otherwise, return $f_2(\pk,a)$.
	\item When $A$ queries $S(\sk,m)$, compute $\SIGN^{O_c}(m,\pk)$.  $\SIGN$ works as follows: compute $(a,b)=O_c(m)$ and return $a$.
	\item When $A$ returns $(m,\sigma)$, compute $\FINISH^{O_c}(m,\sigma,\pk)$.  $\FINISH$ works as follows: compute $(a,b)=O_c(m)$. If $\sigma=a$, abort.  Otherwise, output $(\sigma,a)$.
  \end{MyItemize}
  In addition, we have $\INSTANCE(\pk)=\pk$ and $\START(\INSTANCE(\pk))=(\pk,\pk)$, so $\INSTANCE$ and $\START$ satisfy the required properties.

  \begin{theorem}The reduction above is in \os\ form.\end{theorem}
\begin{proof}
  $\RAND^{O_c}(b,r,\pk)$ is completely random and independently distributed, as $f_1(\pk,a)$ and $f_2(\pk,a)$ are both random ($f_b$ is a permutation and $a$ is truly random).  Observe that $f_1(\pk,\SIGN^{O_c}(m,\pk))\allowbreak =f_1(\pk,a)=O(b,m)$ where $(a,b)=O_c(m)$.  Thus, signing queries are always answered with a valid signature, and the distribution of signatures is identical to that of the correct signing algorithm since $b$ is chosen uniformly.

  Suppose $A$ returns a valid forgery $(m,\sigma)$.  Let $(a,b)=O_c(m)$.  There are two cases, corresponding to whether $\sigma$ corresponds to a signature using $b$ or $1-b$.  In the first case, we have $f_1(\pk,\sigma)=O(b,m)=f_1(\pk,a)$, meaning $\sigma=a$, so we abort.  Otherwise, $f_1(\pk,\sigma)=O(1-b,m)=f_2(\pk,a)$, so $(\sigma,a)$ form a claw.  Since the adversary never asked for a signing query on $m$, these two cases are indistinguishable by the same logic as the proof for FDH.  Thus, the probability of failure is at most a half, which is non-negligible. \qed\end{proof}
  \begin{corollary} If quantum-accessible pseudorandom functions exists, and $(\mathcal{F}_1,\mathcal{F}_2)$ is a pair claw-free trapdoor permutations, then the Katz-Wang signature scheme instantiated with $\mathcal{F}_1$ is secure against quantum adversaries.\end{corollary}

As in the case of FDH, our simulated quantum-accessible random oracle is truly random, so we do not need to rely on Lemma~\ref{oraclerand}.  Thus, the tightness of our reduction is the same as the classical case.  In particular, if the quantum adversary $A_Q$ has  advantage $\epsilon$ then $B$ will have advantage $\epsilon/2$.

\subsubsection*{PSF Signatures from Claw-Free Permutations}

Gentry et al.~\cite{Gentry2008} note that Claw-Free Permutations give rise to pre-image sampleable trapdoor functions (PSFs).  These PSFs can then be used to construct an FDH signature scheme as in Section~\ref{sec:sigsfrompsf}.

Given a pair of claw-free permutations $(\mathcal{F}_1,\mathcal{F}_2)$, define the following PSF: $G$ is just the generator for the pair of permutations.  $\text{Sample}(\pk)$ generates a random bit $b$ and random $x$ in the domain of $f_b$, and returns $(x,b)$.  $f(\pk,x,b)=f_b(\pk,x)$, and $f^{-1}(\sk,y)=(f_b^{-1}(\sk,y),b)$ for a random $b$.  Suppose we have a collision $((x_1,b_1),(x_2,b_2))$ for this PSF.  Then \eq{f_{b_1}(\pk,x_1)=f(\pk,x_1,b_1)=f(\pk,x_2,b_2)=f_{b_2}(\pk,x_2)}  If $b_1=b_2$, then $x_1=x_2$ since $f_{b_1}$ is a permutation.  But this is impossible since $(x_1,b_1)\neq(x_2,b_2)$.  Thus, $b_1\neq b_2$, so one of $(x_1,x_2)$ or $(x_2,x_1)$ is a claw for $(\mathcal{F}_1,\mathcal{F}_2)$.

Hence, we can instantiate FDH with this PSF to get the following signature scheme:
\begin{MyItemize}
	\item $G$ is the generator for the permutations.
	\item $S^O(\sk,m)=(f_b^{-1}(\sk,O(m)),b)$ for a random bit $b$.
	\item $V^O(\pk,m,(\sigma,b))=1$ if and only if $f_b(\pk,\sigma)=O(m)$.
\end{MyItemize}

The security of this scheme follows from Corollary~\ref{cor:psf}, with a similar tightness guarantee (this PSF has only a pre-image min-entropy of 1, which results in a loss of a factor of two in the tightness of the reduction).  In particular, if we have a quantum adversary $A_Q$ for $\mathcal{E}$ with advantage $\epsilon$, we get a quantum algorithm $B_Q$ for the PSF with advantage $\epsilon/2$, which gives us a quantum algorithm $C_Q$ that finds claws of $(\mathcal{F}_1,\mathcal{F}_2)$ with probability $\epsilon/2$.

\section{Encryption Schemes in the Quantum-Accessible Random Oracle Model}\label{enc}

In this section, we prove the security of two encryption schemes.  The first is the BR encryption scheme due to Bellare and Rogaway~\cite{BR1}, which we show is CPA secure.  The second is a hybrid generalization of the BR scheme, which we show is CCA secure.

Ideally, we could define a general type of classical reduction like we did for signatures, and show that such a reduction implies quantum security.  Unfortunately, defining a \os\ reduction for encryption is considerably more complicated than for signatures.  We therefore directly prove the security of two random oracle schemes in the quantum setting.

\subsection{CPA Security of BR Encryption}

In this section, we prove the security of the BR encryption scheme \cite{BR1} against quantum adversaries:

\begin{definition}[BR Encryption Scheme] Let $\mathcal{F}=(G_0,f,f^{-1})$ be an injective trapdoor function, and $O$ a hash function with the same domain as $f(\pk,\cdot)$.  We define the following encryption scheme, $\mathcal{E}=(G,E,D)$ where:
	\begin{MyItemize}
		\item $G=G_0$
		\item $E^O(\pk,m)=(f(\pk,r),O(r)\oplus m)$ for a randomly chosen $r$.
		\item $D^O(\sk,(y,c))=c\oplus f^{-1}(\sk,y)$
	\end{MyItemize}
\end{definition}
A candidate quantum-immune injective trapdoor function can be built from hard problems on lattices~\cite{Peikert2008}.

\begin{theorem}\label{brsecure} If quantum-accessible pseudorandom functions exists and $\mathcal{F}$ is a quantum-immune injective trapdoor function, then $\mathcal{E}$ is quantum CPA secure.
\end{theorem}
We omit the proof of Theorem \ref{brsecure} because the CPA security of the BR encryption scheme is a special case of the CCA security of the hybrid encryption scheme in the next section.

\subsection{CCA Security of Hybrid Encryption}

We now prove the CCA security of the following
standard hybrid encryption, a generalization of the BR encryption scheme scheme \cite{BR1}, built from an injective trapdoor function and symmetric key encryption scheme.

\begin{definition}[Hybrid Encryption Scheme] Let $\mathcal{F}=(G_0,f,f^{-1})$ be an injective trapdoor function, and $\mathcal{E_S}=(E_S,D_S)$ be a CCA secure symmetric key encryption scheme, and $O$ a hash function.  We define the following encryption scheme, $\mathcal{E}=(G,E,D)$ where:
	\begin{MyItemize}
		\item $G=G_0$
		\item $E^O(\pk,m)=(f(\pk,r),E_S(O(r),m))$ for a randomly chosen $r$.
		\item $D^O(\sk,(y,c))=D_S(O(r'),c)$ where $r'=f^{-1}(\sk,y)$
	\end{MyItemize}
\end{definition}

We note that the BR encryption scheme from the previous section is a special case of this hybrid encryption scheme where $\mathcal{E_S}$ is the one-time pad.  That is, $E_S(k,m)=k\oplus m$ and $D_S(k,c)=k\oplus c$.

\begin{theorem}\label{brCCAsecure} If quantum-accessible pseudorandom functions exists, $\mathcal{F}$ is a quantum-immune injective trapdoor function, and $\mathcal{E}_S$ is a quantum CCA secure symmetric key encryption scheme, then $\mathcal{E}$ is quantum CCA secure.\end{theorem}
\begin{proof}
Suppose we have an adversary $A_Q$ that breaks $\mathcal{E}$.  We start with the standard security game for CCA secure encryption:

\medskip\noindent{\bf Game 0.} Define $\myGame_0$ as the game a quantum adversary $A_Q$ plays for problem $\text{Asym-CCA}(\mathcal{E})$.

\medskip\noindent{\bf Game 1.} Define $\myGame_1$ as the following game: the challenger generates $(\sk,\pk)\gets G(1^n)$, a random $r$ in the domain of $\mathcal{F}$, a random $k$ in the key space of $\mathcal{E_S}$, and computes $y=f(\pk,r)$.  The challenger has access to a quantum-accessible random oracle $O_q$ whose range is the key space of $\mathcal{E_S}$.  It then sends $\pk$ to $A_Q$.  The challenger answers queries as follows:
\begin{MyItemize}
	\item Random oracle queries are answered with the random oracle $O_\textrm{\rm quant}$, which takes a basis element $\ket{x,y}$ into $\ket{x,y\oplus O_q(f(\pk,x))}$.
	\item Decryption queries on $(y',c')$ are answered as follows:
	\begin{MyItemize}
		\item[Case 1:] If $y=y'$, respond with $D_S(k,c')$.
		\item[Case 2:] If $y\neq y'$, respond with $D_S(O_q(y'),c')$.
	\end{MyItemize}
	\item The challenge query on $(m_0,m_1)$ is answered as follows: choose a random $b$.  Then, respond with $(y,E_S(k,m_b))$.
\end{MyItemize}
When $A_Q$ responds with $b'$, we say that $A_Q$ won if $b=b'$.

Observe that, because $f$ is injective and $O_q$ is random, the oracle $O_\textrm{\rm quant}$ is a truly random oracle with the same range as $O_q$.  The challenge ciphertext $(y,c)$ seen by $A_Q$ is distributed identically to that of $\myGame_0$.  Further, it is a valid encryption of $m_b$ relative to the random oracle being $O_\textrm{\rm quant}$ if $O_q(y)=k$.  For $y'\neq y$, the decryption of $(y',c')$ is
\eq{D_S(O_q(y'),c')=D_S(O_\textrm{\rm quant}(f^{-1}(\sk,y')),c')=D^{O_\textrm{\rm quant}}(\sk,(y',c'))}
Which is correct.  Likewise, if $O_q(y)=k$, the decryption of $(y,c')$ is also correct.  Thus, the view of $A_Q$ in $\myGame_1$ is identical to that in $\myGame_0$ if $O_q(y)=k$.  We now make the following observations:
\begin{MyItemize}
	\item The challenge query and decryption query answering algorithms never query $O_q$ on $y$.
	\item Each quantum random oracle query from the adversary to $O_\textrm{\rm quant}$ leads to a quantum random oracle query from the challenger to $O_q$.  The query magnitude of $y$ in the challenger's query to $O_q$ is the same as the query magnitude of $r$ in the adversary's query $O_\textrm{\rm quant}$.
\end{MyItemize}
Let $\epsilon$ be the sum of the square magnitudes of $y$ over all queries made to $O_q$ (i.e. the total query probability of $y$).  This is identical to the total query probability of $r$ over all queries $A_Q$ makes to $O_\textrm{\rm quant}$.

We now construct a quantum algorithm $B_\mathcal{F}^{O_q}$ that uses a quantum-accessible random oracle $O_q$, and inverts $f$ with probability $\epsilon/q$, where $q$ is the number of random oracle queries made by $A_Q$.  $B_\mathcal{F}^{O_q}$ takes as input $(\pk,y)$, and its goal is to output $r=f^{-1}(\sk,y)$.  $B_\mathcal{F}^{O_q}$ works as follows:
\begin{MyItemize}
	\item Generate a random $k$ in the key space of $\mathcal{E_S}$.  Also, generate a random $i\in \{1,...,q\}$.  Now, send $\pk$ to $A_Q$ and play the role of challenger to $A_Q$.
	\item Answer random oracle queries with the random oracle $O_\textrm{\rm quant}$, which takes a basis element $\ket{x,y}$ into $\ket{x,y\oplus O_q(f(\pk,x))}$.
	\item Answer decryption queries on $(y',c')$ as follows:
	\begin{MyItemize}
		\item[Case 1:] If $y=y'$, respond with $D_S(k,c')$.
		\item[Case 2:] If $y\neq y'$, respond with $D_S(O_q(y'),c')$.
	\end{MyItemize}
	\item Answer the challenge query on $(m_0,m_1)$ as follows: choose a random $b$.  Then, respond with $(y,E_S(k,m_b))$.
	\item At the $i$th random oracle query, sample the query to get $r'$, and output $r'$ and terminate.
\end{MyItemize}
Comparing our definition of $B_\mathcal{F}^{O_q}$ to $\myGame_1$, we can conclude that the view seen by $A_Q$ in both cases is identical.  Thus, the total query probability that $A_Q$ makes to $O_\textrm{\rm quant}$ at the point $r$ is $\epsilon$.  Hence, the probability that $B_\mathcal{F}^{O_q}$ outputs $r$ is $\epsilon/q$.  If we assume that $\mathcal{F}$ is secure against quantum adversaries that use a quantum-accessible random oracle, then this quantity, and hence $\epsilon$, must be negligible.  As in the case of signatures (Section~\ref{sig}), we can replace this assumption with the assumption that $\mathcal{F}$ is secure against quantum adversaries (i.e. with no access to a quantum random oracle) and that pseudorandom functions exists to reach the same conclusion.

Since $\epsilon$ is negligible, we can change $O_q(y)=k$ in $\myGame_1$, thus getting a game identical to $\myGame_0$ from the adversary's point of view.  Notice that in $\myGame_0$ and $\myGame_1$, $A_Q$ is in a pure state because we are only applying unitary transformations, performing measurements, or performing classical communication.  We are only changing the oracle at a point with negligible total query probability, so Lemma~\ref{oraclechange} tells us that making this change only affects the distribution of the outcome of $\myGame_1$ negligibly.  This allows us to conclude that the success probability of $A_Q$ in $\myGame_1$ is negligibly close to that in $\myGame_0$.

\medskip
Now, assume that the success probability of $A_Q$ in $\myGame_1$ is non-negligible.  We now define a quantum algorithm $B_\mathcal{E_S}^{O_q}$ that uses a quantum-accessible random oracle $O_q$ to break the CCA security of $\mathcal{E_S}$.  $B_\mathcal{E_S}^{O_q}$ works as follows:
\begin{MyItemize}
	\item On input $1^n$, generate $(\sk,\pk)\gets G(1^n)$.  Also, generate a random $r$, and compute $y=f(\pk,r)$.  Now send $\pk$ to $A_Q$ and play the role of challenger to $A_Q$.
	\item Answer random oracle queries with the random oracle $O_\textrm{\rm quant}$, which takes a basis element $\ket{x,y}$ into $\ket{x,y\oplus O_q(f(\pk,x))}$.
	\item Answer decryption queries on $(y',c')$ as follows:
	\begin{MyItemize}
		\item[Case 1:] If $y=y'$, ask the $\mathcal{E_S}$ challenger for a decryption $D_S(k,c')$ to obtain $m'$.  Return $m'$ to $A_Q$.
		\item[Case 2:] If $y\neq y'$, respond with $D_S(O_q(y'),c')$.
	\end{MyItemize}
	\item Answer the challenge query on $(m_0,m_1)$ by forwarding the pair $\mathcal{E_S}$.  When the challenger responds with $c$ (which equals $E_S(k,m_b)$ for some $b$), return $(y,c)$ to $A_Q$.
	\item When $A_Q$ outputs $b'$, output $b'$ and halt.
\end{MyItemize}
Comparing our definition of $B_\mathcal{E_S}^{O_q}$ to that of $\myGame_1$, we can conclude that the view of $A_Q$ in both cases is identical.  Thus, $A_Q$ succeeds with non-negligible probability.  If $A_Q$ succeeds, it means it returned $b$, meaning $B_\mathcal{E_S}^{O_q}$ also succeeded.  Thus, we have an algorithm with a quantum random oracle that breaks $\mathcal{E_S}$.  This is a contradiction if $\mathcal{E_S}$ is CCA secure against quantum adversaries with access to a quantum random oracle, which holds since $\mathcal{E_S}$ is CCA secure against quantum adversaries and quantum-accessible pseudorandom functions exist, by assumption.

Thus, the success probability of $A_Q$ in $\myGame_1$ is negligible, so the success probability of $A_Q$ in $\myGame_0$ is also negligible.  Hence, we have shown that all polynomial time quantum adversaries have negligible advantage in breaking in breaking the CCA security of $\mathcal{E}$, so $\mathcal{E}$ is CCA secure.\qed
\end{proof}

We briefly explain why Theorem~\ref{brsecure} is a special case of Theorem~\ref{brCCAsecure}.  Notice that, in the above proof, $B_\mathcal{E_S}$ only queries its decryption oracle when answering decryption queries made by $A_Q$, and that it never makes encryption queries.  Hence, if $A_Q$ makes no decryption queries, $B_\mathcal{E_S}$ makes no queries at all except the challenge query.  If we are only concerned with the CPA security of $\mathcal{E}$, we then only need $E_S$ to be secure against adversaries that can only make the challenge query.  Further, if we only let $A_Q$ make a challenge query with messages of length $n$, then $E_S$ only has to be secure against adversaries making challenges of a specific length.  But this is exactly the model in which the one-time pad is unconditionally secure.  Hence, the BR encryption scheme is secure, and we have proved Theorem~\ref{brsecure}.

\section{Conclusion}

We have shown that great care must be taken if using the random oracle model when arguing security against quantum attackers. Proofs in the classical case should be reconsidered, especially in case the quantum adversary can access the random oracle with quantum states.  We also developed conditions for translating security proofs in the classical random oracle model to the quantum random oracle model.  We applied these tools to certain signature and encryption schemes.

The foremost question raised by our results is in how far techniques for ``classical random oracles'' can be applied in the quantum case. This stems from the fact that manipulating or even observing the interaction with the quantum-accessible random oracle would require measurements of the quantum states. That, however, prevents further processing of the query in a quantum manner.  We gave several examples of schemes that remain secure in the quantum setting, provided quantum-accessible pseudorandom functions exist.  The latter primitive seems to be fundamental to simulate random oracles in the quantum world. Showing or disproving the existence of such pseudorandom functions is thus an important step.

Many classical random oracle results remain open in the quantum random oracle settings.  It is not known how to prove security of generic FDH signatures as well as signatures derived from the Fiat-Shamir heuristic in the quantum random oracle model.  Similarly, a secure generic transformation from CPA to CCA security in the quantum  RO model is still open.


\section*{Acknowledgments}

We thank Chris Peikert for helpful comments.
Dan Boneh was supported by NSF, the Air Force Office of Scientific Research (AFO\ SR) under a MURI award, and by the Packard Foundation.
Marc Fischlin and Anja Lehmann were supported by grants Fi 940/2-1 and Fi 940/3-1 of the German Research Foundation (DFG).
\"Ozg\"ur Dagdelen and Marc Fischlin were also supported by CASED
(\texttt{www.cased.de}). Christian Schaffner is supported by a NWO
VENI grant.

\bibliographystyle{alpha}
\bibliography{quantum-oracle}

\iffullversion
\appendix

\section{Definitions}
\label{sec:defs}

\begin{definition}[Trapdoor Permutation] A trapdoor permutation is a triple of functions $\mathcal{F}=(G,f,f^{-1})$ where:
	\begin{MyItemize}
		\item $G(1^n)$ generates a private/public key pair $(\sk,\pk)$.
		\item $f(\pk,\cdot)$ is a permutation for all $\pk$.
		\item $f^{-1}(\sk,\cdot)$ is the inverse of $f(\pk,\cdot)$ for all $(\pk,\sk)$ generated by $G$.  That is, $f^{-1}(\sk,f(\pk,x))=x$ and $f(\pk,f^{-1}(\sk,y))=y$.
	\end{MyItemize}
\end{definition}

For a trapdoor permutation $\mathcal{F}$, we define the problem $Inv(\mathcal{F})=(Game(\mathcal{F}),0)$ where $Game(\mathcal{F})$ is the following game between a quantum adversary $A$ and the challenger $Ch$:  $Ch$, on input $n$, runs $G(1^n)$ to obtain $(\sk,\pk)$ and generates a random $y$ in the range of $f(\pk,\cdot)$.  It sends $(\pk,y)$ to $A$.  $A$ is allowed to make quantum random oracle queries $O(\cdot)$.  When $A$ outputs $x$, $Ch$ outputs 1 if and only if $f(\pk,x)=y$.
\begin{definition} A trapdoor permutation $\mathcal{F}$ is secure against quantum adversaries if $Inv(\mathcal{F})$ is hard for quantum computers.\end{definition}

The following definition is due to \cite{Gentry2008}:
\begin{definition}[Preimage Sampleable Trapdoor Function]A quadruple of functions $\mathcal{F}=(G,Sample,f,f^{-1})$ is a trapdoor collision-resistant hash function with preimage sampling (PSF) if:
	\begin{MyItemize}
		\item $G(1^n)$ generates secret and public keys $(\sk,\pk)$.
		\item $f(\pk,\cdot)$ has domain $D$ and range $R$.
		\item $\text{Sample}(1^n)$ samples from a distribution on $D$ such that for all $\pk$ the distribution $f\big(\pk,\ \text{Sample}(1^n)\big)$ is within $\epsilon_{sample}$ of uniform.
		\item $f^{-1}(\sk,y)$ generates an $x$ such that $f(\pk,x)=y$.  The distribution is within $\epsilon_{pre}$ of the conditional distribution of $Sample()$ given $f(\pk,x)=y$, where $\epsilon_{pre}$ is negligible.
		\item Pre-image Min-entropy: For all $y\in R$, the probability of any element in the conditional distribution of $\text{Sample}(1^n)$ given $f(\pk,x)=y$ is less than $\epsilon_{prob}$, where $\epsilon_{prob}$ is negligible.
	\end{MyItemize}
\end{definition}
For a PSF, we define two problems: $Inv(\mathcal{F})$ is identical to the problem with the same name for trapdoor permutations, and $Col(\mathcal{F})=(Game(\mathcal{F}),0)$ where $Game(\mathcal{F})$ is the following game between a quantum adversary $A$ and the challenger $Ch$:  $Ch$, on input $n$, runs $G(1^n)$ to obtain $(\sk,\pk)$, and sends $\pk$ to $A$.  $A$ is allowed to make quantum random oracle queries $O(\cdot)$.  When $A$ outputs a pair $(x_1,x_2)$, $Ch$ outputs 1 if and only if both $x_1\neq x_2$ and $f(\pk,x_1)=f(\pk,x_2)$.
\begin{definition} A PSF $\mathcal{F}$ is secure against quantum adversaries if $Inv(\mathcal{F})$ and $Col(\mathcal{F})$ are both hard for quantum computers.\end{definition}
\cite{Gentry2008} construct a PSF whose security is based on the hardness of lattice problems.

\paragraph{Signature schemes.}
We next review signature schemes using our unified notation.

\begin{definition}[Signature Scheme] A random oracle signature scheme is a triple of functions $\mathcal{S}=(G,S^O,V^O)$ where:
\begin{MyItemize}
	\item $O$ is a random oracle.
	\item $G(1^n)$ generates a pair $(\sk,\pk)$ where $\sk$ is the signer's private key, and $\pk$ is the public key.
	\item $S^O(\sk,m)$ generates a signature $\sigma$.
	\item $V^O(\pk,m,\sigma)$ returns 1 if and only if $\sigma$ is a valid signature on $m$.  \\ We require that $V^O(\pk,m,S^O(\sk,m))=1$ for all $m$ and $(\sk,\pk)$ generated by $G$.
\end{MyItemize}\end{definition}

For a signature scheme $\mathcal{S}$, we define the problem $Sig-Forge(\mathcal{S})=(Game(\mathcal{S}),0)$ where $Game(\mathcal{S})$ is the following game between a quantum adversary $A$ and the challenger $Ch^O$:  $Ch^O$, on input $n$, runs $G(1^n)$ to obtain $(\sk,\pk)$.  It then sends $\pk$ to $A$.  $A$ is allowed to make quantum random oracle queries $O(\cdot)$ and classical signature queries $S^O(\sk,\cdot)$ to $Ch^O$.  When $A$ outputs a forgery candidate $(m,\sigma)$, $Ch^O$ outputs 1 if and only if $A$ never asked for a signature on $m$ and $\sigma$ is a valid signature for $m$ ($V^O(\pk,m,\sigma)=1$).
\begin{definition} A signature scheme $\mathcal{S}$ is secure against quantum adversaries if $Sig-Forge(\mathcal{S})$ is hard for quantum computers.\end{definition}

\paragraph{Encryption.}
We next review encryption systems using our notation.

\begin{definition}[Symmetric Key Encryption Scheme] A symmetric key random oracle encryption scheme is a pair of functions $\mathcal{E}=(E^O,D^O)$ where:
\begin{MyItemize}
	\item $E^O(k,m)$ generates a ciphertext $c$.
	\item $D^O(k,c)$ computes the plaintext $m$ corresponding to ciphertext $c$.  We require that $D^O(k,E^O(k,m))=m$.
\end{MyItemize}\end{definition}
For a symmetric key encryption scheme $\mathcal{E}$, we define the problem $Sym-CCA(\mathcal{E})=(Game(\mathcal{E}),\frac{1}{2})$ where $Game(\mathcal{E})$ is the following game between a quantum adversary $A$ and the challenger $Ch^O$:  $Ch^O$, on input $n$, generates a key $k$ of length $n$ at random, and sends $k$ to $A$.  $A$ is allowed to make quantum random oracle queries $O(\cdot)$, classical encryption queries $E^O(k,\cdot)$, and classical decryption queries $D^O(k,\cdot)$.  $A$ is also allowed one classical challenge query, where it sends $Ch^O$ a pair $(m_0,m_1)$.  $Ch^O$ chooses a random bit $b$, and computes $c=E^O(k,m_b)$, which it sends to $A$.  When $A$ returns a bit $b'$, $Ch^O$ outputs 1 if and only if both $b=b'$ and there was no decryption query $D^O(k,c)$ after the challenge query.

\begin{definition}[Symmetric Key CCA Security] A symmetric key encryption scheme $\mathcal{E}$ is Chosen Ciphertext Attack (CCA) secure against quantum adversaries if $Sym-CCA(\mathcal{E})$ is hard for quantum computers.\end{definition}

\begin{definition}[Asymmetric Key Encryption Scheme] An Asymmetric key encryption scheme is a triple of functions $\mathcal{E}=(G,E^O,D^O)$ where:
\begin{MyItemize}
	\item $G(1^n)$ generates a private/public key pair $(\sk,\pk)$
	\item $E^O(\pk,m)$ generates a ciphertext $c$.
	\item $D^O(\sk,c)$ computes the plaintext $m$ corresponding to ciphertext $c$.  We require that $D^O(\sk,E^O(\pk,m))=m$.
\end{MyItemize}\end{definition}
For a symmetric key encryption scheme $\mathcal{E}$, we define the problem $Asym-CCA(\mathcal{E})=(Game(\mathcal{E}),\frac{1}{2})$ where $Game(\mathcal{E})$ is the following game between a quantum adversary $A$ and the challenger $Ch^O$:  $Ch^O$, on input $n$, uses $G(1^n)$ to generate $(\sk,\pk)$, and sends $\pk$ to $A$.  $A$ is allowed to make quantum random oracle queries $O(\cdot)$ and classical decryption queries $D^O(\sk,\cdot)$.  $A$ if also allowed to make one classical challenge query, where it sends $Ch^O$ a pair $(m_0,m_1)$.  $Ch^O$ chooses a random bit $b$, and computes $c=E^O(\pk,m_b)$, which it sends to $A$.  When $A$ returns a bit $b'$, $Ch^O$ outputs 1 if and only if both  $b=b'$ and there was no decryption query $D^O(\sk,c)$ after the challenge query.

\begin{definition}[Asymmetric Key CCA Security] An Asymmetric key encryption scheme $\mathcal{E}$ is Chosen Ciphertext Attack (CCA) secure against quantum adversaries if $Asym-CCA(\mathcal{E})$ is hard for quantum computers.\end{definition}

\section{Security of the $\ident^*$ Protocol}\label{app:negative}
To prove security of our protocol we need to show that an adversary $\adv$ after interacting with an honest prover $\prov^*$, can subsequently not impersonate $\prov^*$ such that $\veri^*$ accepts the identification.

\parag{Security against Classical Adversaries.}

We first show that the $\ident^*$ protocol is secure in the (standard) random oracle model against classical adversaries and then discuss that there exist hash functions, which can securely replace the random oracle.

\begin{lemma} \label{lem:ro.class}
Let $\ident=(\ikgen,\allowbreak\prov,\veri)$ be a secure identification scheme. Then for any efficient \emph{classical} adversary $\adv$ and $\ell > 6 \log(\alpha)$ the protocol $\ident^*$ is secure in the random oracle.
\end{lemma}
\begin{proof}
Assume towards contradiction that a verifier $\veri^*$ after interacting with an adversary $\adv$, both given $(\pk,\ell)$ as input, accepts with output $b^*=1$. Thus, $\adv$ must have convinced $\veri^*$ in the evaluation of the $\ident$-protocol or provided at least $r/4$ collisions. Due to the independence of the two stages of our protocol (in particular, $\sk$ is not used during the collision search) we have $$\Pr[\adv \text{ ``breaks'' } \ident^*] \leq \Pr[\ccount > r/4] + \Pr[\adv \text{ ``breaks'' } \ident].$$
Since we assume that the underlying identification scheme is secure, the latter probability is negligible. Thus, it remains to show that an adversary $\adv$ with access to a random oracle $H$ finds $r/4$ near-collisions on $H(k_i,\cdot)$ for given $k_i$ in time $O(\sqrt[3]{2^\ell})$ with negligible probability only. In the random oracle model, the ability of finding collisions is bounded by the birthday attack, which states that after sending $\sqrt{2^\ell}$ random input values\footnote{Note that we give all statements for a random oracle outputting directly $\ell\leq\log(n)$ bits, as we are interested in near-collisions. Such an oracle can be obtained from a random oracle with range $\bin^n$ by simply truncating the output to the first $\ell$ bits.}, at least one pair will collide 
with probability $\geq 1/2$. Taking possible parallel power of the adversary into account, the protocol allows $\adv$ to make at most $\alpha\cdot\sqrt[3]{2^\ell}$ queries for some constant $\alpha \geq 1$
(Assumption~\ref{assu:parallel}).
Since $\ell > 6 \log(\alpha)$ we have $\alpha\cdot\sqrt[3]{2^\ell} < \sqrt{2^\ell}$ and thus $\adv$'s success probability for finding a collision in each round is $< 1/2$ which vanishes when repeating the collision search $r$ times.

More concretely, the upper bound on the birthday probability for $q$ queries and a function with range size $N$ is given by $\frac{q(q-1)}{2N}$ (see e.g.~\cite{BelKilRog94}).
Thus, when considering an adversary making at most $q=\alpha\sqrt[3]{2^\ell}$ queries to a random oracle with range
$\bin^\ell$ we obtain:
\[ \Pr[\coll] ~\leq~ \frac{\alpha^2}{2\sqrt[3]{ 2^\ell}}~\leq~ \frac{\alpha^2}{2\sqrt[3]{n}}
\]
due to the choice of $\ell \leq \log{n}$. The repetition of such a constrained collision search does not increase the success probability of the adversary, since the verifier sends a fresh ``key'' $k_i$ in each round. Thus, the adversary cannot reuse already learned values from the random oracle, but has to start the collision search from scratch for each new key. That is, the probability of $\adv$ finding a collision is at most $\Pr[\coll]$ in each round.

Applying the Chernoff-bound yields the probability for finding at least $r/4$ collision in $r$ independent rounds:
\[ \Pr[\ccount > r/4] ~\leq~ \exp\left(-\frac{r\alpha^2}{2\sqrt[3]{n}} \cdot \left(\frac{\sqrt[3]{n}-2\alpha^2}{2\alpha^2}\right)^2 \cdot \frac{1}{4} \right) ~\leq~ \exp\left(-\frac{r\sqrt[3]{n}}{32\alpha^2}\right)
\]

Thus, for a constant $\alpha$, and setting $r = \text{poly}(n)$ the above term is negligible in $n$. However, then, the overall success probability of $\adv$ is negligible as well.
\end{proof}

When considering classical adversaries only, we can securely instantiate the random oracle in the $\ident^*$ scheme by a hash function $H$ that provides near-collision-resistance close to the birthday bound. Under this assumption, the security proof of our identification scheme carries over to the standard model, as well. (We omit a formal proof, as it follows the argumentation of Lemma~\ref{lem:ro.class} closely.) Note that it is a particular requirement of the SHA-3 competition~\cite{sha3}, that the hash function candidates achieve collision-resistance approximately up to the birthday bound and provide this property also for any fixed subset of the hash functions' output bits. Thus, all remaining SHA-3 candidates (or at least the winner of the competition) is supposed to be quasi-optimal near-collision-resistant.


\parag{Security against Quantum Adversaries.}
We now show that such a result is impossible in the quantum world, i.e., for any hash function $H$ there exists a quantum-adversary $\adv_Q$ that breaks the $\ident^*$ protocol (regardless of the security of the underlying identification scheme). This contrasts with the security that can still be achieved in the (classical) random oracle model:

\begin{lemma}
Let $\ident_Q=(\ikgen,\allowbreak\prov,\veri)$ be a secure quantum-immune identification scheme. Then for any efficient \emph{quantum} adversary $\adv_Q$ and $\ell > 6 \log(\alpha)$ the protocol $\ident^*$ is secure in the random oracle model.
\end{lemma}
\begin{proof}
By assuming that $\ident_Q$ is a quantum-immune identification scheme, an adversary $\adv_Q$ trying to convince a verifier $\veri^*$ in the $\ident^*$ protocol must provide at least $r/4$ many collisions in the first stage of the protocol. Thus, we have to show that a quantum adversary $\adv_Q$ can succeed in the collision-search with negligible probability only.

Note that in order to gain advantage of the quantum speed-up (e.g., by applying Grover's search algorithm) the random oracle $H$, resp. the indicator function based on $H$, has to be evaluated on quantum states, i.e., on superpositions of many input strings.  However, by granting $\adv_Q$ only classical access to the random oracle, it is not able to exploit its additional quantum power to find collisions on $H$. Thus, $\adv_Q$ has to stick to the classical collision-search on a random oracle, which we have proven to succeed in $r/4$ of $r$ rounds with negligible probability, due to the constraint of making at most $\alpha\cdot\sqrt[3]{2^\ell}$ oracle queries per round (see proof of Lemma~\ref{lem:ro.class} for details).
\end{proof}


We now show that our  $\ident^*$ scheme becomes totally insecure for any instantiation of the random oracle by a hash function $\hash$.
\newcommand{\hqeval}{\hash_Q.\textsf{Eval}}
\begin{lemma}\label{lem:hash.quant}
There exist an efficient \emph{quantum} adversary $\adv_Q$ such that for any hash function $\hash=(\hkgen,\heval)$  the protocol $\ident^*$ is \emph{insecure}.
\end{lemma}
\begin{proof}
For the proof, we show that a quantum-adversary $\adv_Q$ can find collisions on $\hash$ in at least $r/4$ rounds with non-negligible probability. To this end, we first transform the classical hash function $\hash$ into a quantum-accessible function $\hash_{Q}$.  For the transformation, we use the fact that any classical computation can be done on a quantum computer as well~\cite{ChaNie00}. 
The ability to evaluate $\hash_Q$ on superpositions then allows to apply Grover's algorithm in a straightforward manner: for any key $k_i$ that is sent by the verifier $\veri^*$, the adversary invokes Grover's search on an indicator function testing whether $\hqeval(k_i,x)|_\ell =  \hqeval(k_i,x')|_\ell$ for distinct $x \neq x'$ holds. After $\sqrt[3]{2^\ell}$ evaluations of $\hash_{Q}$ the algorithm outputs a collision $M_i,M'_i$ with probability $>1/2$. As we assume that a quantum evaluation of $\hash_Q$ requires roughly the same time than an evaluation of the corresponding classical function $\hash$, and we do not charge $\adv_Q$ for any other computation, the collision search of $\adv_Q$ terminates before $\veri^*$ stops a round of the collision-finding stage. 

Hence, $\adv_Q$ provides a collision with probability $>1/2$ in each of the $r$ rounds. Using the Chernoff bound, we can now upper bound the probability that $\adv_Q$ finds \emph{less} than $r/4$ collision as:
\[ \Pr[\ccount < r/4] ~\leq~ \exp\left(-\frac{r}{2} \cdot \left(\frac{1}{2}\right)^2 \cdot \frac{1}{2} \right) ~\leq~ \exp\left(-\frac{r}{16}\right)
\]
which is roughly $ \Pr[\ccount < r/4] ~\leq~ 0.94^r $ and thus negligible as a function of $r$. That is, the adversary $\adv_Q$ can make $\veri^*$ accept the interaction with noticeable probability at least $1-\Pr[\ccount < r/4]$.
\end{proof}

As Grover's algorithm only requires (quantum-accessible) black-box access to the hash function, the approach described in the proof of Lemma~\ref{lem:hash.quant} directly applies to the quantum-accessible random oracle model, as well:

\begin{lemma}
The protocol $\ident^*$ is \emph{not} secure in the quantum-accessible random oracle model.
\end{lemma}

\section{Proof of Lemma~\ref{oraclerand}}

Before we prove Lemma~\ref{oraclerand}, we need to prove the following two technical lemmas:
\begin{lemma} Let $|\phi\rangle$ and $|\phi'\rangle$ be superpositions with $|\phi-\phi'|\leq\gamma$.  Let $P$ be some property on strings.  Suppose measuring $|\phi\rangle$ gives a string that satisfies $P$ with probability $\epsilon$.  Then measuring $|\phi'\rangle$ gives a string that satisfies $P$ with  probability $\epsilon'$ where \eq{\sqrt{\epsilon}-\gamma\leq \sqrt{\epsilon'}\leq  \sqrt{\epsilon}+\gamma}\end{lemma}

\begin{proof}
We will prove this lemma geometrically.  We can think of a state $\ket{\phi}$ as a vector $\vec{\phi}$ in $\mathbb{C}^n$.  Then the basis elements $\ket{x}$ as elements of the standard basis for $\mathbb{C}^n$.  We are given that $|\phi-\phi'|\leq\gamma$, meaning that $\vec{\phi}$ and $\vec{\phi'}$ have Euclidean distance of at most $\gamma$.

For a bit string $x$, the probability that sampling $\ket{\phi}$ results in $x$ is $\magn{\Angle{\vec{x},\vec{\phi}}}^2$.  Let $S_P$ be the set of basis elements $\vec{x}$ such that $x$ satisfied $P$.  The probability that sampling $\ket{\phi}$ results in a string satisfying $P$ is then given by \eq{\displaystyle\sum_{\vec{x}\in S_P} \magn{\Angle{\vec{x},\vec{\phi}}}^2}

This also is the square of the length of the projection of $\vec{\phi}$ onto the subspace spanned by $S_P$.  So, let $\vec{\phi_P}$ and $\vec{\phi_P'}$ be the projections of $\vec{\phi}$ and $\vec{\phi'}$ onto the space spanned by $S_P$.  The probability that sampling $\ket{\phi}$ (resp. $\ket{\phi'}$) results in a string satisfying $P$ is simply $\epsilon=\magn{\vec{\phi_P}}^2$ (resp. $\epsilon'=\magn{\vec{\phi_P'}}^2$).  Projections only decrease distance, so by the triangle inequality,

\eq{\sqrt{\epsilon'}=\magn{\vec{\phi_P'}}\leq \magn{\vec{\phi_P}}+\magn{\vec{\phi_P}-\vec{\phi_P'}} \leq \magn{\vec{\phi_P}}+\magn{\vec{\phi}-\vec{\phi'}} \leq \sqrt{\epsilon}+\gamma}

Reversing the roles of $\ket{\phi}$ and $\ket{\phi'}$ gives us the other inequality.\end{proof}

\begin{lemma}Let $A$ be an quantum algorithm that makes at most $q$ queries to quantum random oracle $O$.  Fix a $y$ in the co-domain of $O$.  The expected value of the total query probability of all $x$ such that $O(x)=y$ is at most $\frac{2 q^3}{2^m}$.  \end{lemma}
\begin{proof}
	Suppose we have an oracle $O'$ for which the output on every input is distributed identically and independently, with a uniform distribution over $\{0,1\}^m\setminus\{y\}$.  We now modify the oracle as follows: for each input $x$, with probability $2^{-m}$, replace the output with $y$.  This oracle is now a random oracle, so its distribution is identical to $O$.
	
	Let $\sigma_i$ be the total query magnitude over the first $i-1$ queries of $x$ such that we change $O'(x)$.  Let $\delta_i$ be the query magnitude of those $x$ in the $i$th query.  Let $\gamma_i$ be the Euclidean distance between the state of $A$ at the $i$th query when using oracle $O'$ and the modified oracle $O$.  By (\ref{oraclechange}), $\gamma_i\leq\sqrt{(i-1)\sigma_i}$.  Let $\rho_i$ be the query magnitude of $x$ such that $O(x)=y$ (which is the same as the query probability of $x$ such that we changed $O'(x)$).  By the above lemma,
	\begin{eqnarray*}
	\rho_i&\leq&(\sqrt{\delta_i}+\gamma_i)^2\\
	&=&\delta_i+\gamma_i^2+2\sqrt{\delta}\gamma_i\\
	&\leq&\delta_i+(i-1)\sigma_i+2\sqrt{(i-1)\delta_i\sigma_i}\\
	&\leq&\delta_i+(i-1)\sigma_i+2\sqrt{i-1}(\delta_i+\sigma_i)
	\end{eqnarray*}
	Now, observe that since we are deciding whether to change the output of a query point at random and independently, the expected query probability of the points that we changed in each query is exactly $2^{-m}$.  Thus, $\E\brackets{\delta_i}=2^{-m}$ and $\E\brackets{\sigma_i}=(i-1)2^{-m}$.  Thus,
	\eq{\E\brackets{\rho_i}\leq 2^{-m} (1+(i-1)^2+2\sqrt{i-1}(1+(i-1)))\leq 2^{-m}2 i^2}.
	This result is not surprising, as it implies that any quantum algorithm which is to output a preimage of $y$ with overwhelming probability must make $O(\sqrt{2^{-m}})$ quantum oracle queries, which is well known lower bound for the unstructured search problem (see Bennett et al.~\cite{Bennett1997} for more).
	Summing over all $q$ queries gives the expected query probability of $x$ such that $O(x)=y$ to be at most $2\times2^{-m} q^3$.
	\end{proof}
	
	\begin{proofof}{Lemma~\ref{oraclerand}}
		We are given a random oracle $O$ and a distribution $D$ that is $\epsilon$-close to uniform.  Observe that:
		\begin{eqnarray*}
			\epsilon&=&\sum_y \magn{\Pr[y|D]-2^{-m}}\\
			&=&\sum_{y:\Pr[y|D]\geq 2^{-m}}\parens{\Pr[y|D]-2^{-m}}+\sum_{y:\Pr[y|D]<2^{-m}}\parens{2^{-m}-\Pr[y|D]}\\
			0&=&\sum_{y:\Pr[y|D]\geq 2^{-m}}\parens{\Pr[y|D]-2^{-m}}-\sum_{y:\Pr[y|D]<2^{-m}}\parens{2^{-m}-\Pr[y|D]}
		\end{eqnarray*}
		Thus, \eq{\frac{\epsilon}{2}=\sum_{y:\Pr[y|D]\geq 2^{-m}}\parens{\Pr[y|D]-2^{-m}}=\sum_{y:\Pr[y|D]<2^{-m}}\parens{2^{-m}-\Pr[y|D]}}
		
		Define a distribution $D'$ as follows:
		\begin{MyItemize}
			\item If $\Pr[y|D]<2^{-m}$, $\Pr[y|D']=0$.
			\item If $\Pr[y|D]\geq 2^{-m}$, $\Pr[y|D']=(\Pr[y|D]-2^{-m})2/\epsilon$
		\end{MyItemize}
		All the probabilities are clearly non-negative.  For this to be a probability distribution, the probabilities need to um to 1:
		\eq{\sum_y \Pr[y|D']=\sum_{y:\Pr[y|D]\geq 2^{-m}}(\Pr[y|D]-2^{-m})\frac{2}{\epsilon}=\frac{\epsilon}{2}\frac{2}{\epsilon}=1}
		
Now, we can create another distribution $D''$ as follows: first, generate $y$ uniformly at random.  Then,
		\begin{MyItemize}
			\item If $\Pr[y|D]\geq 2^{-m}$, output $y$.
			\item If $\Pr[y|D]<2^{-m}$, then with probability $2^m \Pr[y|D]$, output $y$.  Otherwise, pick a $y'$ from $D'$ and output $y'$.
		\end{MyItemize}
		If $\Pr[y|D]<2^{-m}$, $\Pr[y|D'']=2^{-m}\times(2^m \Pr[y|D])=\Pr[y|D]$.  Otherwise,
		\begin{eqnarray}
			\Pr[y|D'']&=&2^{-m}+\sum_{y':\Pr[y'|D]<2^{-m}}2^{-m}(1-2^m\Pr[y'|D])\Pr[y|D']\\
			 &=&2^{-m}+\sum_{y':\Pr[y'|D]<2^{-m}}(2^{-m}-\Pr[y'|D])(\Pr[y|D]-2^{-m})\frac{2}{\epsilon}\\
			&=&2^{-m}+\frac{\epsilon}{2}(\Pr[y|D]-2^{-m})\frac{2}{\epsilon}=\Pr[y|D]
		\end{eqnarray}
		Thus $D''=D$.  This demonstrates that we can construct the oracle $O'$ whose elements are distributed according to $D$ as follows: Start with the random oracle $O$, and for each input $x$, if $\Pr[O(x)|D]<2^{-m}$, then with probability $1-2^m \Pr[O(x)|D]$, replace the output with a $y'$ drawn from $D_X'$.  Otherwise leave the oracle unchanged at that point.
		
		Now we bound the expected query magnitude of $x$ such that the oracle changed.  By the above lemma, the expected total query probability of any $x$ such that $O(x)=y$ is $2 q^3 2^{-m}$.  Let $\sigma$ be the query magnitude of points $x$ at which we changed the oracle:		
		
		\begin{eqnarray*}
		\E[\sigma]&=&\E\brackets{\sum_{x:\Pr[O(x)|D]<2^{-m}}(1-2^{-m}\Pr[O(x)|D])\times(\text{total query magnitude of }x)}\\
		&=&\sum_{y:\Pr[y|D]<2^{-m}}(1-2^m\Pr[y|D])\E[\text{total query magnitude of $x$ such that $O(x)=y$}]\\
		&\leq&\sum_{y:\Pr[y|D]<2^{-m}}(1-2^m\Pr[y|D]) 2 q^3 2^{-m}\\
		&=&2 q^3\sum_{y:\Pr[y|D]<2^{-m}}(2^{-m}-\Pr[y|D])=\frac{2 q^3 \epsilon}{2}=q^3\epsilon
		\end{eqnarray*}
		Thus the expected Euclidean distance is \eq{\E[\sqrt{q\sigma}]\leq\sqrt{q\E[\sigma]}\leq\sqrt{q\times q^3 \epsilon}=q^2\sqrt{\epsilon}}
		This means the expected variational distance of the output distributions is at most $4q^2\sqrt{\epsilon}$.  Thus, the distribution of outputs when the oracle values are distributed according to $D$ is at most $4q^2\sqrt{\epsilon}$ away from the distribution of outputs when the oracle is truly random.
\end{proofof}
\fi
\end{document}